\RequirePackage{fix-cm}
\documentclass[a4,10pt]{article}
\usepackage{graphicx}
\usepackage{latexsym}

\usepackage{xspace}
\usepackage{url}
\usepackage{wrapfig}
\usepackage[T1]{fontenc}
\usepackage{slantsc}
\usepackage[]{epsfig}
\usepackage{cite}                    
\usepackage{hyperref}
\usepackage{subfigure}
\usepackage{epstopdf}
\usepackage{amsmath,amssymb,amsfonts,amsfonts,amsthm}

\newtheorem{theorem}{Theorem}
\newtheorem{lemma}[theorem]{Lemma}
\newtheorem{proposition}[theorem]{Proposition}
\newtheorem{observation}[theorem]{Observation}
\newtheorem{corollary}[theorem]{Corollary}
\newtheorem{definition}[theorem]{Definition}

\newcommand{\email}[1]{\href{mailto:#1}{#1}}

\def\capStyle#1{#1}

\newcommand{\calelize}[1]{\ensuremath{\mathcal #1}\xspace}
\newcommand{\calA}{\calelize{A}}

\newcommand{\calS}{\calelize{S}}

\newcommand{\calTS}{\calelize{T(S)}}
\newcommand{\calTSi}{\calelize{T(S_i)}}

\newcommand{\calAS}{\calelize{A(S)}}

\newcommand{\TrpSrchTree}{\ensuremath{\mathbb{T}}\xspace}
\newcommand{\G}{\ensuremath{\mathbb{G}}\xspace}
\newcommand{\Tall}{\ensuremath{\calelize{T}^\ast}\xspace}
\newcommand{\Rall}{\ensuremath{\calelize{R}^\ast}\xspace}
\newcommand{\Arr}[1]{\ensuremath{\calA}(#1)\xspace}
\newcommand{\tred}{\ensuremath{t}\xspace} 
\newcommand{\rred}{\ensuremath{r}\xspace} 
\newcommand{\TcollReduc}{\ensuremath{\calelize{T}^{c}}\xspace} 
\newcommand{\RcollReduc}{\ensuremath{\calelize{R}^{c}}\xspace}  

\newcommand{\E}{\ensuremath{\mathbb{E}}\xspace}

\newcommand{\Left}{\text{\em left}}
\newcommand{\Top}{\text{\em top}}
\newcommand{\Bottom}{\text{\em bottom}}
\newcommand{\Right}{\text{\em right}}

\newcommand{\myqed}{}

\newcommand{\providelength} [1]{
  \ifthenelse{\isundefined{#1}}
             {\newlength{#1}}
             {}}

\providelength{\subfigwidth}
\providelength{\spacewidth}

\newcommand{\ie}{i.\,e.,\xspace}

\newcommand{\cpp}{\textsc{C\raise.08ex\hbox{\tt ++}}\xspace}
\newcommand{\cgal}{{\sc Cgal}\xspace}

\newcommand{\RIC}{random incremental construction\xspace}
\newcommand{\NN}{nearest-neighbor\xspace}
\newcommand{\PL}{point location\xspace}
\newcommand{\PPL}{planar point location\xspace}
\newcommand{\lqpl}{\ensuremath{{\mathcal L}}\xspace}
\newcommand{\depth}{\ensuremath{{\mathcal D}}\xspace}

\newcommand{\N}{\ensuremath{\mathbb{N}}}

\newcommand{\deberg}{de Berg\xspace}

\newcommand{\first}[2]{#1}

\newcommand{\ignore}[1]{}
\first
{\newcommand{\TODO}[1]{}}
{\newcommand{\TODO}[1]{\medskip\begin{tabular}{rp{5cm}}{\bf TODO}&{\bf #1}\medskip\end{tabular}}}
\newcommand{\arxivdcg}[2]{#2} 

\begin{document}


\title{Optimal randomized incremental construction for guaranteed logarithmic planar point location%
\thanks{
This work has been supported in part by the 7th Framework
Programme for Research of the European Commission, under
FET-Open grant number 255827 (CGL---Computational Geometry
Learning), by the Israel Science Foundation (grant no.
1102/11), by the German-Israeli Foundation (grant no.
1150-82.6/2011), and by the Hermann Minkowski--Minerva
Center for Geometry at Tel Aviv University.
}}

\author{
  Michael Hemmer%
  \thanks{ 
    Institute of Operating Systems and Computer Networks,
    University of Technology Braunschweig, Braunschweig, Germany;
    \email{mhsaar@gmail.com}}%
  \and 
  Michal Kleinbort 
  \thanks{
    School of Computer Science, Tel Aviv University, Tel Aviv, Israel;
    \email{balasmic@tau.ac.il}, \email{danha@tau.ac.il}
  }
  \and 
  Dan Halperin$^\ddagger$}

\date{21 October, 2014}

\maketitle

\begin{abstract}
Given a planar map of $n$ segments in which we wish to
efficiently locate points, we present the first randomized
incremental construction of the well-known trapezoidal-map
search-structure that only requires expected
$O(n \log n)$ preprocessing time while deterministically guaranteeing
worst-case linear storage space and worst-case
logarithmic query time.
This settles a long standing open problem; the best
previously known construction time of such a structure,
which is based on a directed acyclic graph, so-called the {\it history DAG},
and with the above worst-case space and query-time guarantees,
was expected $O(n \log^2 n)$.
The result is based on a deeper understanding of the structure
of the history DAG,
its depth in relation to the length of its longest search path,
as well as its correspondence to the
{\it trapezoidal search  tree}.
Our results immediately extend to planar maps induced by finite collections of
pairwise interior disjoint well-behaved curves.\\

The article significantly extends the theoretical aspects of the work presented in \url{http://arxiv.org/abs/1205.5434}.

\end{abstract}
\clearpage


 \section{Introduction}
\label{sec:intro}


The \PPL problem for a set $S$ of $n$
pairwise interior disjoint $x$-monotone curves
inducing a planar subdivision (or a planar arrangement)~$\calA(S)$  is defined as follows:
given a query point~$q$, locate the feature of~$\calA(S)$
containing~$q$, i.e., the face, edge or vertex of~$\calA(S)$  that~$q$ lies in.
It is one of the fundamental problems in Computational
Geometry and has numerous applications in a variety of
domains, such as computer graphics, motion planning,
computer aided design (CAD), geographic information
systems (GIS), and many more.

In this work we revisit one of the most elegant and general
algorithms for planar point location,
namely the randomized incremental construction
of the trapezoidal map and the related search structures.
It is also the only algorithm for general $x$-monotone curves
that has an exact, complete and maintained
implementation~\cite{FlatoHHNE2000_planar_maps_cgal,HemmerKH12},
which is available via \cgal, the Computational Geometry
Algorithms Library~\cite{cgal}.

\subsection{Previous Work}

As a core problem in Computational Geometry,
the \PPL problem has been
well-studied for many years.
Among the various solutions to the problem,
some methods can only provide an \emph{expected} query time of $O(\log{n})$
but cannot guarantee logarithmic query time for all cases.
It is particularly true for solutions that only require $O(n)$ space.
In addition, certain solutions may only support linear subdivisions,
while others are applicable to non-linear ones as well.
Triangulation-based \PL methods, such as Kirkpatrick's
approach~\cite{Kirkpatrick1983_opt_stch_subdivisions} and
Devillers's Delaunay Hierarchy~\cite{Devillers2002_del_hierarchy}
are restricted to linear subdivisions, since they build on
a triangulation of the actual input.
Kirkpatrick creates a hierarchy of $O(\log{n})$ levels of triangulated faces
(including the outer face),
where at each level an independent-set of low-degree vertices
is removed when creating the next level in the hierarchy.
This approach guarantees that the data structure
requires only $O(n)$ space and that a query
takes only $O(\log n)$ time.
The \emph{Delaunay Hierarchy} of Devillers, on the other hand,
does not guarantee logarithmic query time,
and may have a linear query time in the worst case.

\arxivdcg{
The \emph{Landmarks} \PL strategy~\cite{HaranH08}
is a method maintaining good running times in practice,
which is available for \cgal{} arrangements.
The Landmarks strategy combines a \NN search to find the nearest landmark,
and a hopefully short walk in the full subdivision
from the landmark to the query point.
It is a heuristics,
and therefore does not guarantee a logarithmic query time.
}{}

Most of the other methods can be summarized under the
\emph{trapezoidal search graph} model of computation,
as pointed out by Seidel and Adamy~\cite{SeidelA2000_wc_query_comp}.
The fundamental search structure
used by this model is a directed acyclic graph~\G
(which may even be just a tree for some methods)
with one root and many leaves.
Internal nodes in~\G have two outgoing
edges each, and are either labeled with a vertical line
and are therefore left-right nodes, or labeled by an
input curve and in such a case are top-bottom nodes.
In principal,
all these solutions can be generalized to support
well-behaved curves~\cite[Subsection~1.3.3]{FogelHW2012_CGAL_arr_book},
that is, curves that can be decomposed into a finite number
of $x$-monotone pieces.

One of the earliest solutions that can be subsumed under this model
is known as the \emph{slabs method} introduced by
Dobkin and Lipton~\cite{DobkinL1976_multidim_srch_prblm}.
Every endpoint induces a vertical wall
giving rise to $2n+1$ vertical slabs.
A point location query is performed by
a binary search to locate the correct slab
and another search within the slab in $O(\log n)$ time.
Preparata~\cite{Preparata1981_trpz_graph_pl} introduced the
\emph{Trapezoid Graph method} based on the slabs method.
His method, reduces the space bounds from
$O(n^2)$, as required by Dobkin and Lipton's slabs method,
to $O(n\log{n})$ only, by uniquely decomposing each edge
into $O(\log{n})$ fragments.
Sarnak and Tarjan~\cite{SarnakT1986_pl_perst_trees}
achieved a significant improvement in memory usage
for a slabs-based method
by using a~\emph{persistent} data structure.
Their key observation is that the sequence of search structures
in all slabs can be interpreted as one structure that changes over time,
which can be stored as a persistent data structure
requiring only $O(n)$ size.
Another example for this model is the \emph{separating
chains method} by Lee and Preparata~\cite{LeeP1976_sep_chains_pl},
which also requires linear space.
Their algorithm expects a monotone subdivision
and uses horizontal monotone chains to separate faces.
It is based on the idea that faces
of any monotone subdivision can be totally ordered
preserving the above-below relation.
Each chain is a node in a binary search tree
(each edge is kept only once).
Querying the structure is essentially deciding whether
the query point is above or below $O(\log n)$ chains.
However, for each chain this test takes $O(\log n)$, using
a binary search.
Therefore, the total query time is $O(\log ^2 n)$.
Another linear size data structure was proposed by
Edelsbrunner et al.~\cite{EdelsbrunnerGS1986_opt_pl_mon_subdiv}.
They used Fractional Cascading in order to create a layered chain tree as
a search structure by copying every other $x$-value from a node to its parent
and maintaining pointers from parent list to child lists.
Querying this structure takes $O(\log n)$ time.

This work is focused on the trapezoidal map
randomized incremental construction (RIC), which was introduced by
Mulmuley~\cite{Mulmuley1990_fast_planar_part_alg}
and Seidel~\cite{Seidel1991_simp_fast_inc_rand_td}.
Its associated search structure is a Directed Acyclic Graph (DAG)
recording the history of the construction.
It achieves expected $O(n \log n)$ preprocessing time,
expected $O(\log n)$ query time and expected $O(n)$ space.
As pointed out by \deberg et al.~\cite{CG-alg-app},
the latter two can even be guaranteed. However,
their sketched solution would require $O(n \log^2 n)$ time.
%
%
A general major advantage of all variants of this approach is
that they can also handle dynamic scenes to some extent,
namely, it is possible to add or delete edges later on.
The entire method is discussed in more detail in
Section~\ref{sec:prelim} below.

In an invited talk~\cite{Seidel_cccg_2009} at CCCG 2009, Raimund Seidel
briefly sketched a deterministic variant with equivalent guarantees that,
like the approach of Kirkpatrick,
uses independent sets to determine a proper insertion order of segments.
However, he also concludes that
the elegant and less cumbersome randomized approach, i.e., the RIC,
is preferable.

A variant of the latter
adds weights and thus gives expected query time
satisfying entropy bounds~\cite{AryaMM2001_simp_entrpy_alg}.
Arya et al.\ also stated that entropy preserving cuttings can be
used to give a method the query time of which 
approaches the optimal entropy bound,
at the cost of increased space and
programming complexity~\cite{AryaMM2001_entrpy_cutt}.
These methods guarantee a logarithmic query time,
however maintaining the search structures
requires a considerably large amount of memory and
a significant increase in the preprocessing time.
Therefore, these solutions are generally
rather complicated to implement.
For other methods
and variants the reader is referred to a comprehensive overview
given in~\cite{Snoeyink2004_handbook_disc_comp_geom_pl}.

\subsection*{Contribution}

This article extends the theoretical aspects of the work presented
in the European Symposium on Algorithms (ESA) 2012~\cite{HemmerKH12},
which also presented a major revamp of the exact implementation of the RIC,
which now guarantees $O(\log n)$ query time and $O(n)$ space.
%
%

Section~\ref{sec:prelim} discusses the basic algorithm
by Mulmuley~\cite{Mulmuley1990_fast_planar_part_alg}
and Seidel~\cite{Seidel1991_simp_fast_inc_rand_td}
and the variant by \deberg et al.~\cite{CG-alg-app}.
The latter guarantees logarithmic query time
by reconstructing the search structure
if the length of the longest search path~\lqpl
or the size~\calS exceed some thresholds.
However, to keep the preprocessing efficient,~\lqpl
and~\calS would have to be efficiently
accessible, which is not trivial for~\lqpl.
In fact, an early version of~\cite{CG-alg-app}
did not make the distinction between~\lqpl and the
depth~\depth of the DAG, which is the length
of the longest DAG path and can be efficiently accessed.
Section~\ref{sec:depth_vs_qlen} discusses
the fundamental difference between~\depth and~\lqpl.
Specifically, we show that the worst-case ratio between~\depth and~\lqpl
can be $\Theta(n/\log n)$.
%
In Section~\ref{sec:eff_ver_alg} we introduce two algorithms to
verify~\lqpl after the actual construction, both leading to an overall
expected $O(n \log n)$ preprocessing time.
The first relies on a deeper understanding of the relation
between the trapezoidal search tree \TrpSrchTree and the DAG \G
(A preparatory discussion of~\TrpSrchTree and~\G is given in Section~\ref{sec:relation_G_and_T}).
It operates directly on \G and
requires expected $O(n \log n)$ time.
The second runs in deterministic $O(n \log n)$ time and
is based on the computation of the ply of all trapezoids
that existed during the construction of \G.
Conclusions and open problems
are given in Section~\ref{sec:conclusion}.
We defer to Appendices some auxiliary material including
some straightforward case-analysis, description of folklore
results that we have not found archived, and adaptation of known
algorithms to our specific needs.

\section{Preliminaries}
\label{sec:prelim}%

This section briefly reviews the trapezoidal map \RIC
for \PL.
After some relevant general definitions in
Subsection~\ref{subsec:background:defs},
the basic algorithm presented by
Mulmuley~\cite{Mulmuley1990_fast_planar_part_alg}
and Seidel~\cite{Seidel1991_simp_fast_inc_rand_td}
is provided in Subsection~\ref{subsec:background:basic}.
The variant by \deberg et al.~\cite{CG-alg-app},
which gives the guarantees on \lqpl and \calS, is described
in Subsection~\ref{subsec:background:guaranteed}.


\subsection{Definitions}
\label{subsec:background:defs}

\arxivdcg{
\begin{figure}
\centering
  \includegraphics[width=0.7\textwidth]{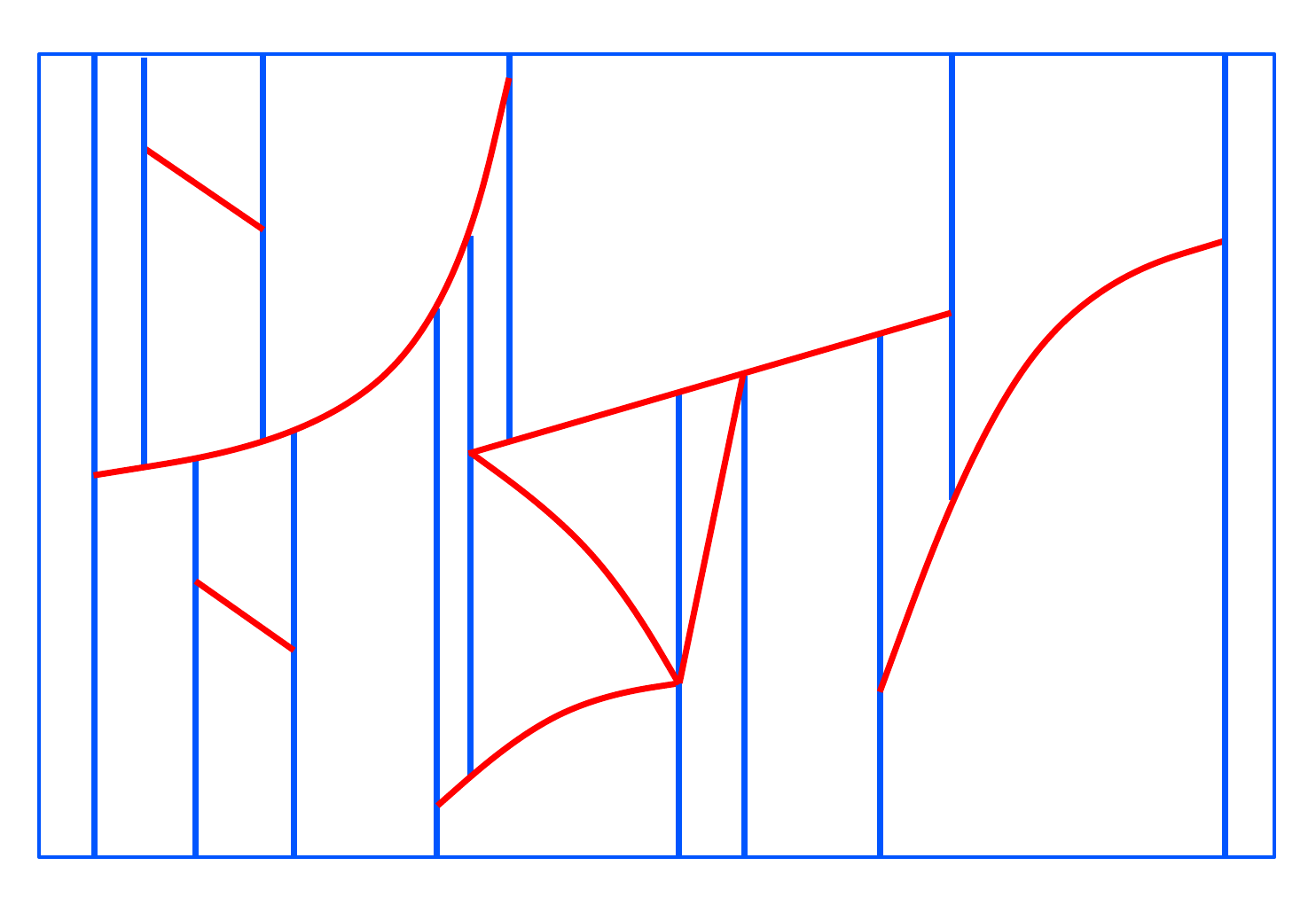}
  \caption[Trapezoidal Map Illustration]
   {\capStyle{
    The trapezoidal map \calTS for an arrangement \calAS.
    The edges of \calAS are in red while the vertical walls
    are in blue.}}
  \label{fig:trapezoidal_map_def}
\end{figure}
}{}

Let $S$ be a set of $n$ pairwise interior disjoint
$x$-monotone curves in general position,
i.e., no two distinct endpoints have the same $x$-coordinate and
no endpoint of one curve lies in the interior of another curve.
$S$ induces a planar subdivision
(or a planar arrangement) $\calA(S)$, which
is composed of vertices, and faces,
in addition to its $n$ edges.

The \emph{Trapezoidal Map} of an arrangement \calAS,
denoted by \calTS, is obtained by
extending vertical walls
from each endpoint upwards and downwards
until an input curve is reached
or the wall extends to infinity.
\begin{wrapfigure}{r}{0.32\textwidth}
\centering
  \vspace{-10pt}
  \includegraphics[width=0.32\textwidth]{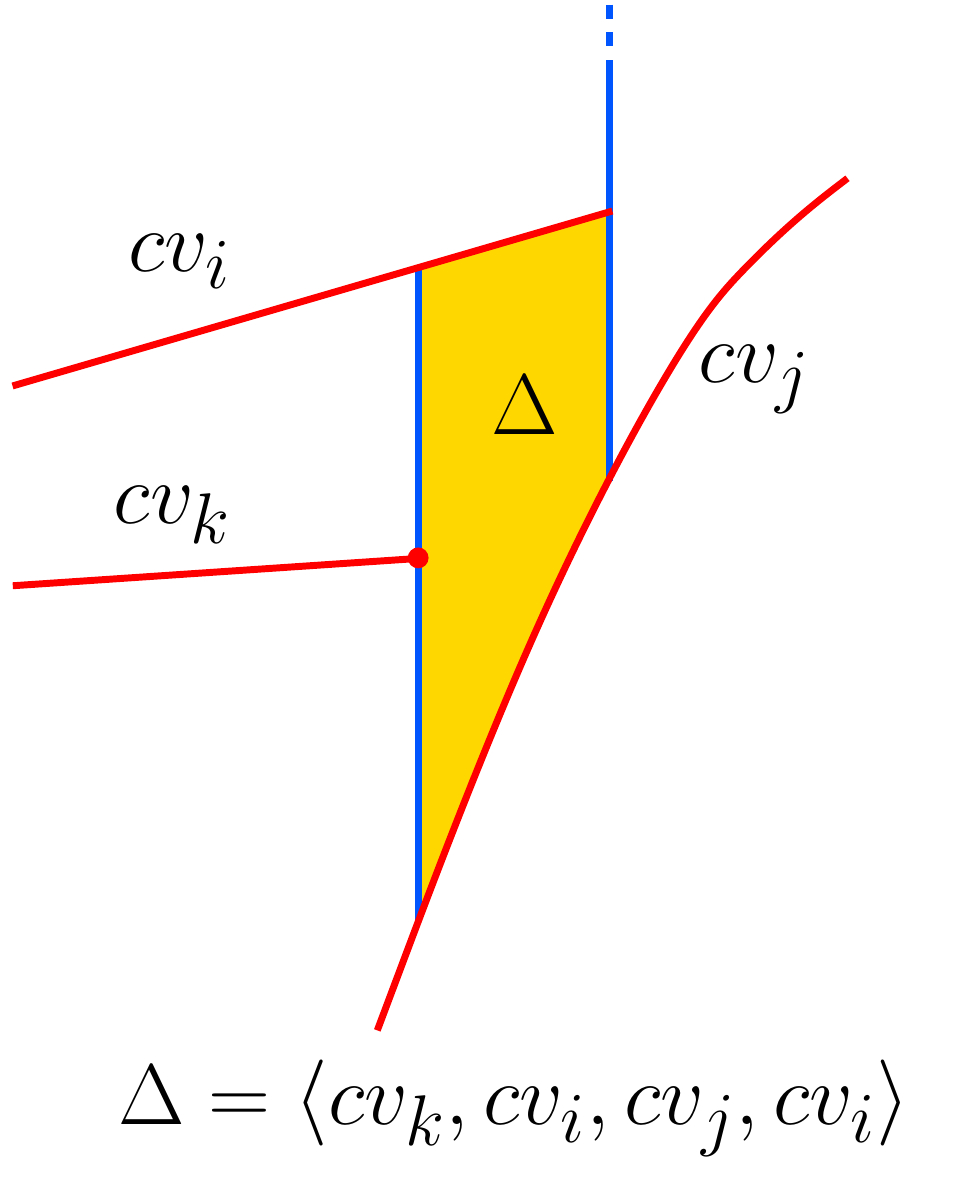}
  \vspace{-10pt}
  \label{fig:trapezoid}
  \vspace{-10pt}
\end{wrapfigure}
Each trapezoid $\Delta$ of \calTS is restricted by at most two
curves%
\footnote{We use the term {\it trapezoid} even when the side edges are
not linear segments.},
denoted by~\Bottom($\Delta$) and~\Top($\Delta$),
and also by either one or two vertical walls (trapezoid bases).
\Left($\Delta$), \Right($\Delta$) denote the curves
whose endpoints induce the left and right vertical walls, respectively.
Therefore, each trapezoid $\Delta$ in \calTS can be defined
by a unique quadruplet:
$\langle$\Left($\Delta$), \Right($\Delta$), \Bottom($\Delta$), \Top($\Delta$)$\rangle$,
as depicted in the figure to the right.
%
\arxivdcg{In a degenerate trapezoid, having only one vertical edge (a triangle),
the two non-vertical edges are adjacent.}{}
The trapezoidal map \calTS is unique and does not depend on the order of insertion.
As shown in~\cite{CG-alg-app}, \calTS of
an arrangement consisting of $n$
curves
has at most $3n + 1$ trapezoids and at most
$6n + 4$ vertices.
\arxivdcg{In other words, the trapezoidal map of an arrangement of linear size
has a linear size as well.}{}

\arxivdcg{
The trapezoids in~\calTS are neighbors if they share a vertical wall.
As we assume general position, each trapezoid has at most four neighboring
Each trapezoid has at most four neighboring
trapezoids: at most two along its left vertical edge,
and the same along its right vertical edge.
In order to avoid the general position assumption, a symbolic shear
transformation should be used.
This is done by employing
lexicographical comparison, that is,
comparing first by the $x$-coordinate and then by the $y$-coordinate.
This implies that two covertical points produce
a virtual trapezoid, which has zero width.

For simplicity of presentation, and w.l.o.g., 
the following figures contain horizontal line-segments.
However, our analysis is for general $x$-monotone curves.
}
{The trapezoids in~\calTS are neighbors if they share a vertical wall.
As we assume general position, each trapezoid has at most four neighboring
trapezoids: at most two along its left vertical edge,
and the same along its right vertical edge.
We remark that the general position assumption poses
no restriction on the algorithm as one
can use a symbolic shear transformation, i.e., by simply replacing comparisons
of $x$-coordinates by lexicographical $xy$-comparisons.

For simplicity of presentation, and w.l.o.g., 
the following figures contain horizontal line-segments.
}

\subsection{The Basic RIC Algorithm}
\label{subsec:background:basic}

We review the random incremental construction (RIC)
of a \PL structure, as introduced
in~\cite{Mulmuley1990_fast_planar_part_alg,Seidel1991_simp_fast_inc_rand_td}
and described in~\cite{CG-alg-app,CG-intro-rand-alg}.
Given an arrangement \calAS of $n$ pairwise interior disjoint $x$-monotone
curves, a random permutation of the curves is inserted incrementally,
constructing the trapezoidal map \calTS.
During the incremental construction,
an auxiliary search structure, a directed acyclic graph (DAG) \G,
is maintained.
The DAG \G has one root and many leaves, one for every trapezoid
in the trapezoidal map \calTS.
Every internal node is a binary decision node, representing either
an endpoint $p$ of an input curve, deciding whether a query point $q$ lies
to the left or to the right of the vertical line through $p$,
or an $x$-monotone curve $cv_i$, deciding whether the query point $q$ is above or below it.
When reaching a curve-node representing a curve $cv_i$, it is guaranteed that
the query point $q$ lies in the $x$-range of $cv_i$.
%
%
In addition, the trapezoids in the leaves of \G are interconnected,
that is, each trapezoid knows its (at most) four horizontal neighbors
(two to the left and two to the right).
\arxivdcg{
In particular, there are no common $x$-coordinates for two distinct
endpoints, since we use a symbolic shear transform.
For example, trapezoid $A$ in Figure~\ref{fig:background:dag_insert}(a),
has only two neighbors, namely $B$ and $C$, which are
its top-right and bottom-right neighbors, respectively.
Trapezoid $C$ has also two different neighbors,
which are trapezoid $A$ to its left and trapezoid $D$ to its right.
In the underlying data structure, each trapezoid maintains four
pointers to its potential neighbors.
}{}
For a simple example refer to Figure~\ref{fig:background:dag_insert}(a),
where the DAG is still a tree.

\arxivdcg{
\subsubsection{Querying}
\label{subsubsec:background:basic:querying}
%
%
%
In order to locate the trapezoid containing a query point $q$ one needs to
perform a search in the history DAG.
The search begins at the root of the DAG
and ends in a DAG leaf.
At each internal node a decision is made,
depending on whether the node represents an endpoint or a curve,
in order to find the next node in the path.
In other words, the number of comparisons used by one query equals
to the number of internal nodes along a search path to the
leaf representing the trapezoid that contains the query point.
Therefore, the query time is proportional to the length of the path.
}{}

\begin{figure} [h] 
  \centering
    \setlength{\subfigwidth}{0.51\textwidth}

  \subfigure[]{
        \hspace{-3mm}
        \includegraphics[width=\subfigwidth]{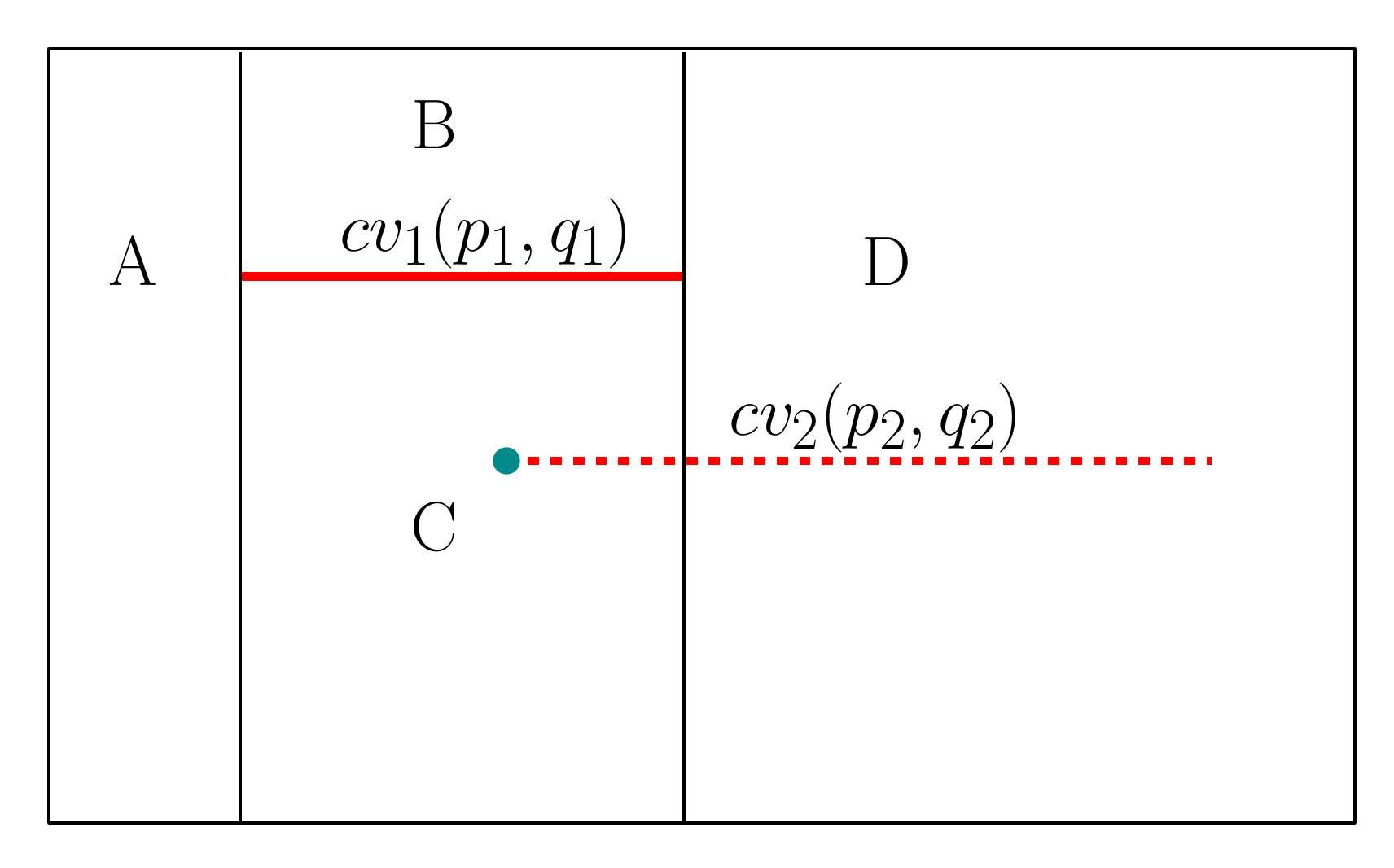}\hspace{-2mm}
        \includegraphics[width=\subfigwidth]{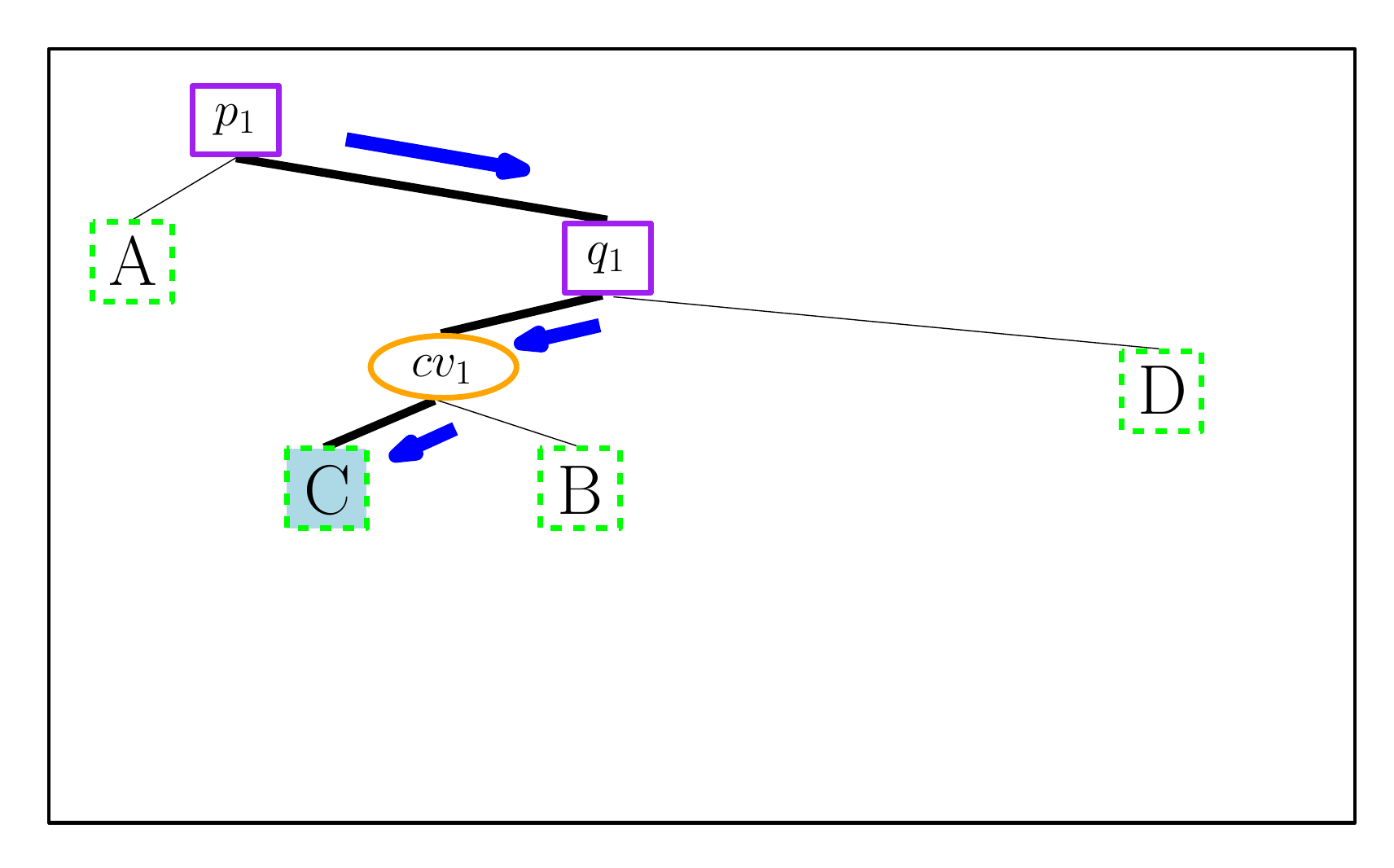}
  }

  \subfigure[]{
        \hspace{-3mm}
      \includegraphics[width=\subfigwidth]{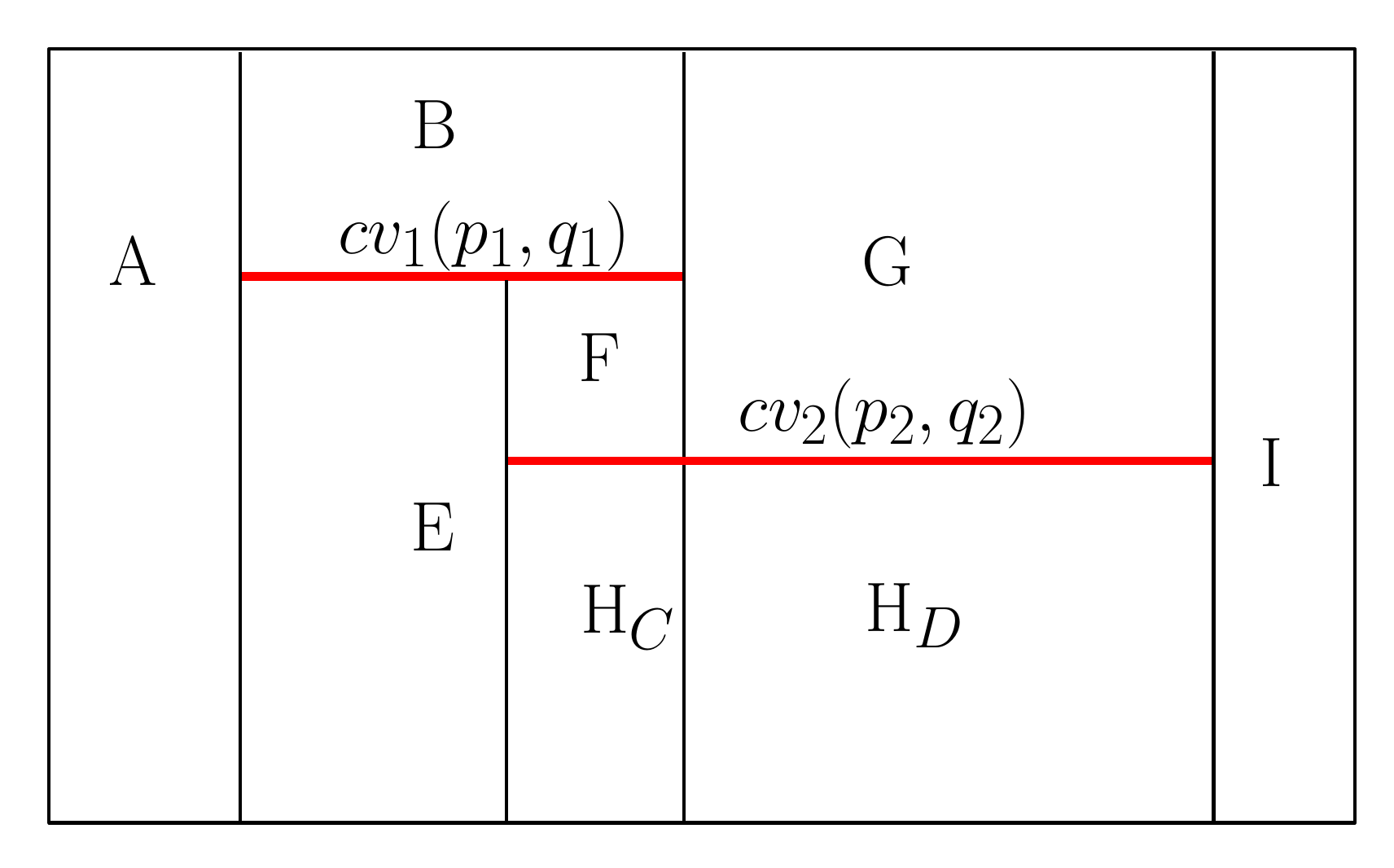}\hspace{-2mm}
      \includegraphics[width=\subfigwidth]{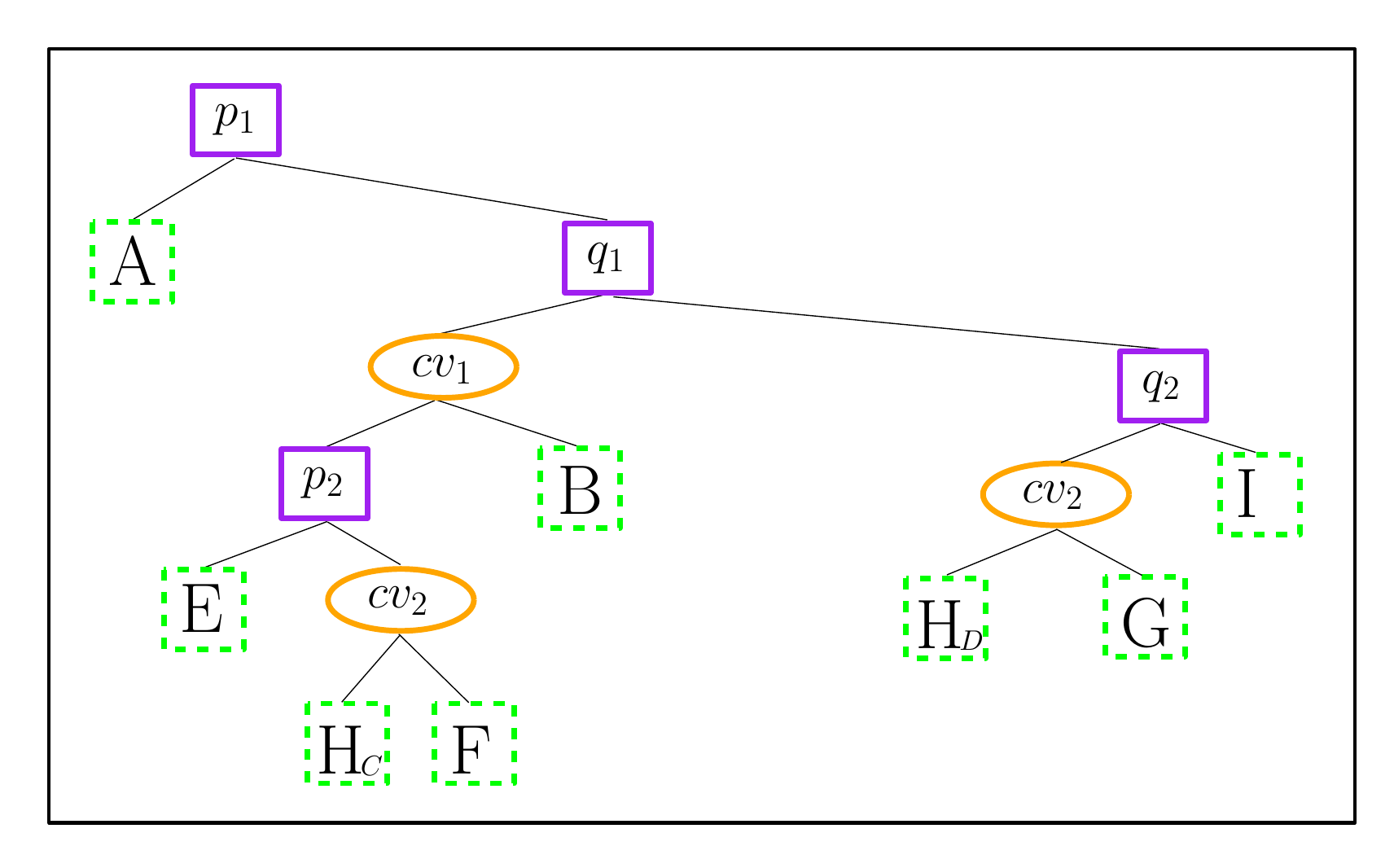}
  }

\subfigure[]{
    \hspace{-3mm}
      \includegraphics[width=\subfigwidth]{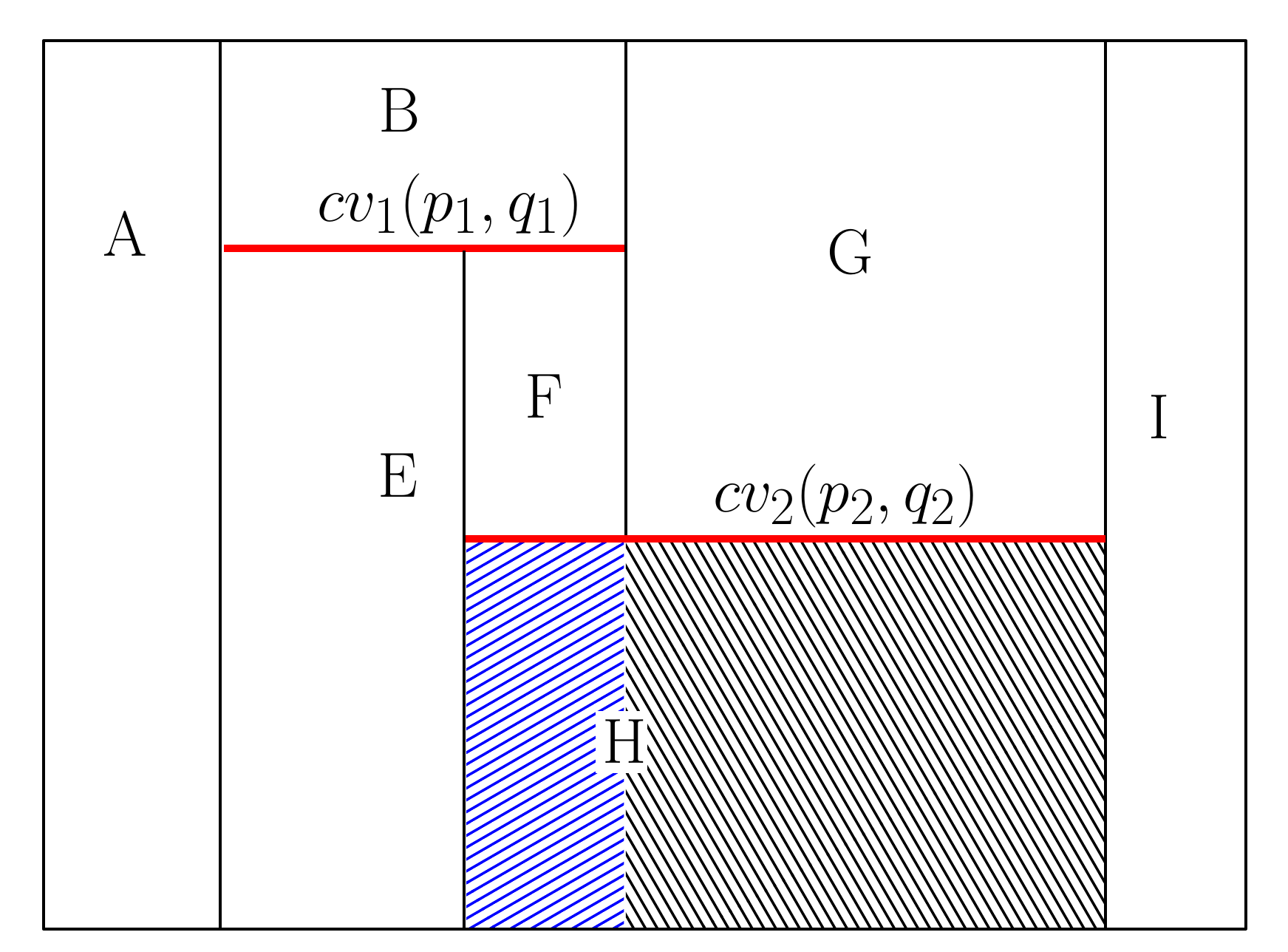}\hspace{-2mm}
      \includegraphics[width=\subfigwidth]{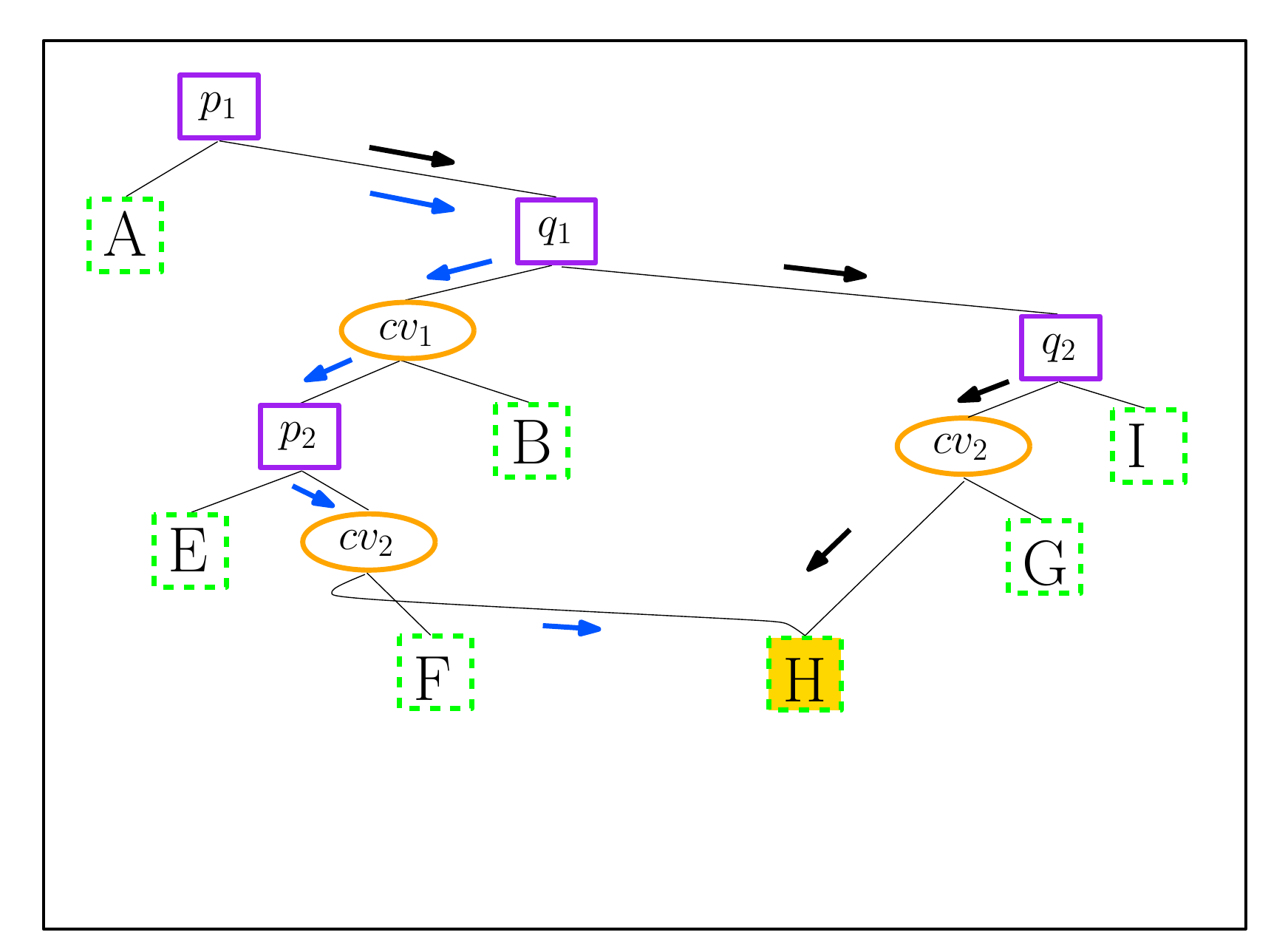}
  }

  \vspace{-10pt}
   \caption[Updating the Search Structure]{
   \capStyle{Updating the DAG with a second curve $cv_2(p_2,q_2)$.
      (a) Locating $p_2$ (the left endpoint of $cv_2$)
      in the DAG of the trapezoidal map for $cv_1(p_1,q_1)$.
      The query path is highlighted with blue arrows.
      (b) Due to the insertion of $cv_2$, trapezoids $C$ and $D$ are split
      into trapezoids ${E, F, H_C}$ and ${H_D, G, I}$, respectively.
      The obtained structure is in fact the Trapezoidal Search Tree~\TrpSrchTree for $cv_1, cv_2$.
      (c) Newly created trapezoids $H_C, H_D$ are merged into trapezoid~$H$,
      since $cv_2$ blocks the wall induced by~$q_1$.
      The resulting DAG contains two directed paths (marked) from the root to the leaf
      labeled~$H$, which are both valid search paths.
}}
  \label{fig:background:dag_insert}
\end{figure}

\subsubsection{Insertion}
\label{subsubsec:background:basic:insertion}
When a new $x$-monotone curve is inserted,
the trapezoid containing its left endpoint is located
by a search from root to leaf.
Then, using the connectivity information described above,
the trapezoids intersected by the curve are
gradually revealed and updated.
Merging new trapezoids, if needed,
takes time that is linear in the number of intersected trapezoids.
The merge turns the data structure into a DAG
with expected $O(n)$ size~\cite{Mulmuley1990_fast_planar_part_alg,Seidel1991_simp_fast_inc_rand_td}.
By skipping the merge step, one would obtain a binary tree,
known as the \emph{trapezoidal search tree}, having  expected
$O(n\log{n})$ size, as shown in Section~\ref{sec:relation_G_and_T}.
The whole insertion process is illustrated in Figure~\ref{fig:background:dag_insert}.
For an unlucky insertion order the size of the resulting data structure
may be quadratic, and the longest search path may be linear.
However, since the curves are inserted in a random order, one can expect $O(n)$ space,
$O(\log{n})$ query time, and $O(n\log{n})$ preprocessing time.
For proofs see~\cite{Mulmuley1990_fast_planar_part_alg,Seidel1991_simp_fast_inc_rand_td,CG-alg-app}.

As a result of the merge 
the search structure may contain more than one valid search
path from the root to a certain leaf,
as demonstrated in Figure~\ref{fig:background:dag_insert}(c).

\subsection{Previous Attempts at Guaranteeing Logarithmic Query Time}
\label{subsec:background:guaranteed}

The basic algorithm, described in Subsection~\ref{subsec:background:basic},
requires expected $O(n\log n)$ time and expected $O(n)$ space.
Moreover, it is not hard to see that the expected query time for
an arbitrary but fixed query point is $O(\log n)$.
However, \deberg et al.~\cite{CG-alg-app} showed that the
probability that the length of the longest search path~\lqpl is
larger than $3\lambda\ln(n+1)$ is rather small, e.g., for~$\lambda=20$ and~$n>4$
it is less than $1/4$.
A similar argument can be applied for the size~\calS of the constructed DAG~\cite{Kleinbort_master}.
This leads to the idea that one can guarantee a linear size data structure
with guaranteed logarithmic query time by simply re-constructing the data
structure with a new random insertion order until
it has the required properties.
Essentially \deberg et al.~\cite{CG-alg-app} show the following crucial lemma:

\begin{lemma}
\label{lemma:prelim:guaranteed:constant_rebuilds}
It is possible to choose suitable constants $c_1, c_2 > 0$ such that
the expected number of rebuilds
that are required to achieve
$\calS < c_1 n$ and $\lqpl < c_2 \log n$
is a small constant.
\end{lemma}

Not taking the cost for the verification of~\lqpl and \calS
into account this would immediately lead to an algorithm that
in total still runs in expected $O(n\log n)$ time;
see also Section~\ref{sec:eff_ver_alg}.
However, while it is straightforward to keep track of~\calS,
an efficient verification of~\lqpl is not trivial at all.
Specifically, one should be aware of that~\lqpl
is not equal to the depth \depth of the DAG, i.e.,
the length of the longest DAG path.
Note that~\lqpl is determined only by the valid search paths and not by all paths.
This subtle difference, which is discussed in Section~\ref{sec:depth_vs_qlen},
caused some confusion in the past.
Thus, since the sketch of the verification algorithm in the last version of
\deberg et al.~\cite{CG-alg-app} requires $O(n\log^2 n)$ time,
until now no expected $O(n\log n)$ time algorithm
giving the above guarantees was known.

 \section{Depth vs. Maximum Query Path Length}
\label{sec:depth_vs_qlen}

The depth~\depth of the DAG
is an upper bound on~\lqpl, as the set of all possible
search paths is a subset of all paths in the DAG.
\depth can be made accessible in constant time,
by storing the depth of each leaf in the leaf itself, and by maintaining
the maximum depth in a separate variable.
The cost of maintaining the depth can be charged to new nodes,
since existing nodes never change their depth value.
Having said that, it is not clear how to efficiently access~\lqpl while retaining linear space,
since each leaf would have to store a non-constant number of values, i.e.,
one for each valid search path that reaches it.
In fact, the memory consumption would be equivalent to the trapezoidal search tree,
which is expected to be of size $O(n\log{n})$%
, as shown in Appendix~\ref{appendix_trapezoidal_search_tree}.
\arxivdcg{In practice there is a considerable difference between the size of the
resulting search tree and the size of the resulting DAG, as demonstrated in
Table~\ref{tbl:tree_dag_size}.
This indicates that performing merges is truly necessary.
\begin{table}
  \vspace{-10pt}
  \caption{
    The table displays the number of trapezoidal search tree nodes
    vs.\ number of DAG nodes for the
    same input with the same insertion order.
    The trapezoidal search tree is obtained by inserting the curves
    in the same order as in the construction of the DAG,
    however, while avoiding merges.
    The last column presents the ratio
    between the number of tree nodes and the number of
    DAG nodes.
    As expected, its values correspond to the function $\log n$,
    where $n$ is the number of edges.}
  \begin{center}
    \begin{tabular}{| c | c | c | c |}
    \hline
    \# Arrangement Edges & \# Tree nodes & \# DAG nodes & \ ratio\\
    \hline
    22 & 125 & 101 & 1.23\\
    \hline
    138 & 1263 & 681 & 1.85\\
    \hline
    285 & 3167 & 1492 & 2.12\\
    \hline
    1483 & 23511 & 8019 & 2.93\\
    \hline
    2977 & 51551 & 16330 & 3.15\\
    \hline
    14975 & 350629 & 84576 & 4.14\\
    \hline
    29973 & 759075 & 169355 & 4.48\\
    \hline    \end{tabular}
  \end{center}
  \label{tbl:tree_dag_size}
  \vspace{-20pt}
\end{table} 
}{}

We show that the depth~\depth of a given DAG
can be linear while its maximum query path length~\lqpl
is still logarithmic, that is, such a
DAG would trigger an unnecessary rebuild.
\arxivdcg{
It is thus questionable whether
we can still expect a constant number of
rebuilds when relying on~\depth.
Our experiments, reported in Subsection~\ref{subsec:depth_vs_qlen:experiments},
show that in practice the two values hardly differ,
indicating that it is sufficient to rely on~\depth.
However, a theoretical proof to consolidate this is still missing.
}
{
It is thus questionable whether
we can still expect a constant number of
rebuilds when relying on~\depth.
}
Figure~\ref{fig:depth_vs_qlen:wc_ratio:third_curve}
demonstrates the difference between the DAG
depth~\depth and the maximum query path length~\lqpl.

\begin{figure}[t] 
  \centering
  \setlength{\subfigwidth}{0.51\textwidth}
  \hspace{-3mm}
    \includegraphics[width=\subfigwidth]{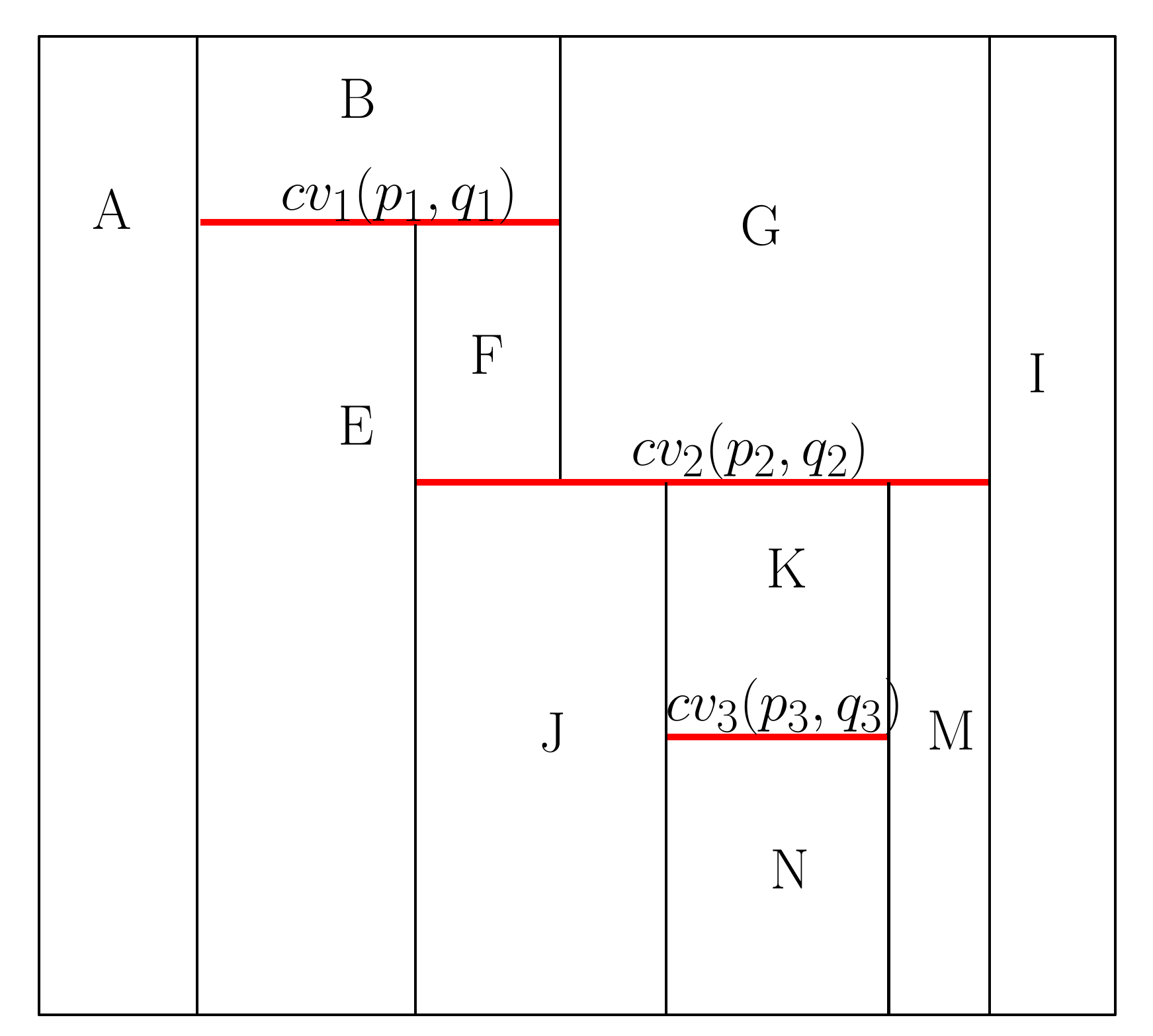} \hspace{-2mm}
    \includegraphics[width=\subfigwidth]{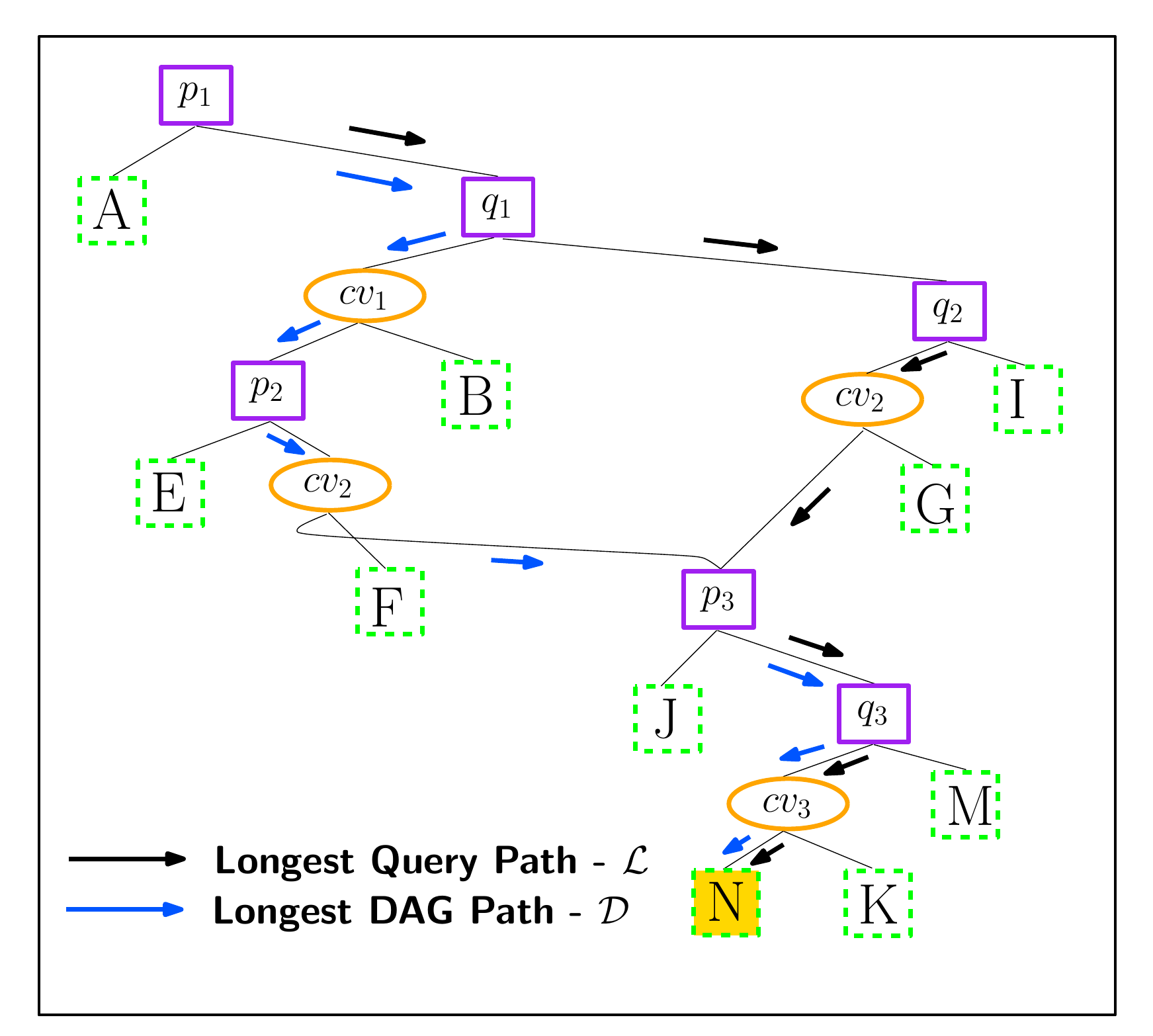}

  \caption[The Longest Search Path vs. the Longest DAG Path]
  {\capStyle{The trapezoidal map and corresponding DAG after inserting
      $cv_3(p_3, q_3)$ to the subdivision from Figure~\ref{fig:background:dag_insert}.
      The leaf representing trapezoid $H$ in the former
      structure is replaced with a subtree rooted at $p_3$.
      There are two directed paths starting at the root that reach
      trapezoid~$N$, marked in black and blue.
      The black
      represents the longest query path,
      and the blue represents the longest DAG path.
      The black path
      is the search path for all queries that end up in trapezoid $N$.
      The blue path is not a valid search path, since
        all points in $N$ are to the right of~$q_1$,
        that is, such a query
        would never visit the left child of $q_1$.
    This scenario occurs due to the merge
    that was part of the insertion of $cv_2$
    (see Figure~\ref{fig:background:dag_insert}(c))
    creating two different paths to a leaf,
    which became the inner node $p_3$ in the updated structure.
      }}
  \label{fig:depth_vs_qlen:wc_ratio:third_curve}
\end{figure}

\arxivdcg{
In Subsection~\ref{subsec:depth_vs_qlen:wc_ratio:sqrt_construction},
a construction in which the difference
between~\depth and~\lqpl is significant is presented.
However, the ratio between the two values can be even
larger, as shown in
Subsection~\ref{subsec:depth_vs_qlen:wc_ratio:wc_construction}.
In particular, the $\Omega(n/\log{n})$
ratio of~\depth and~\lqpl achieved in the latter
is tight. 
}
{}

%

\arxivdcg{
\subsection{Towards a Worst Case Bound} 
\label{subsec:depth_vs_qlen:wc_ratio:sqrt_construction}

\begin{wrapfigure}{r}{0.4\textwidth}
\centering
\includegraphics[width=0.4\textwidth]{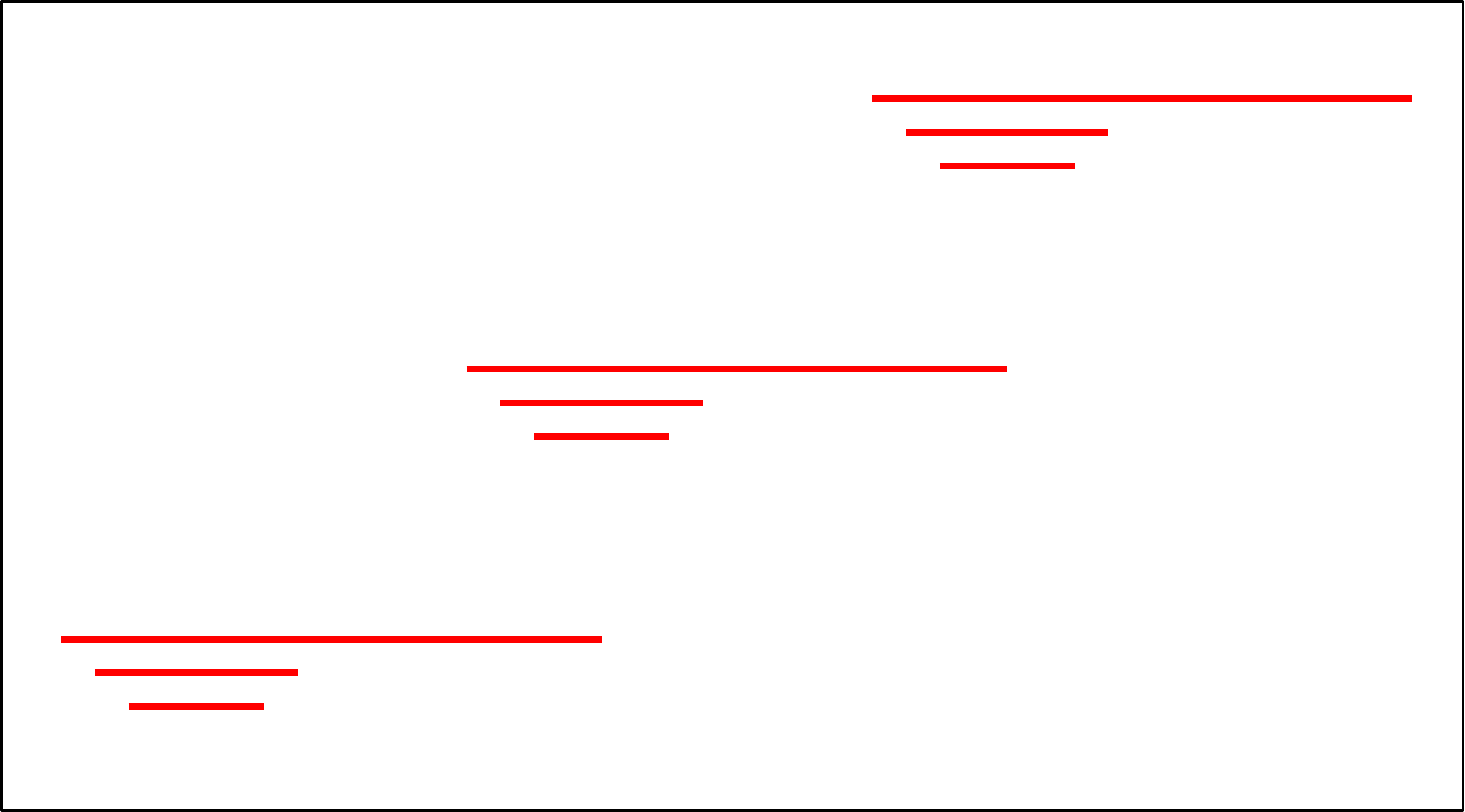}
\vspace{-0.6cm}
\end{wrapfigure}
The figure to the right
demonstrates a simple construction achieving
an $\Omega(\sqrt{n})$ lower bound for the
ratio between \depth and \lqpl.
Assuming that $n = k^2 \in \N$, the construction consists of $k$ blocks,
each containing $k$ horizontal segments.
The blocks are arranged as depicted in the figure.
Segments are inserted from top to bottom.
A block starts with a large segment at the top, which we call the {\it cover segment},
while the other segments successively shrink in size.
Now the next block is placed to the left and below the previous block.
\arxivdcg{Only the cover segment of this block extends below the previous block,
which causes a merge as illustrated in Figure~\ref{fig:depth_vs_qlen:wc_ratio:sqrt_detail}.}
{Only the cover segment of this block extends below the previous block,
which causes a merge.}
All $k=\sqrt{n}$ blocks are placed in this fashion.
This construction ensures that each newly inserted segment intersects the
trapezoid with the largest depth, which increases~\depth.
The largest depth of $\Omega(n)$ is finally achieved in the
trapezoid below the lowest segment.
However, the actual search path into this trapezoid
has only $O(\sqrt{n})$ length,
since
there are $O(\sqrt{n})$ blocks and
a query is able to skip an entire block using only
one comparison to the leftmost point of the cover segment.
Within the relevant block there are at most $O(\sqrt{n})$ possible comparisons.

\begin{figure}[h] 
  \centering
  \setlength{\subfigwidth}{0.45\textwidth}
  \begin{tabular}{cc}
    \subfigure[]{\includegraphics[width=\subfigwidth]{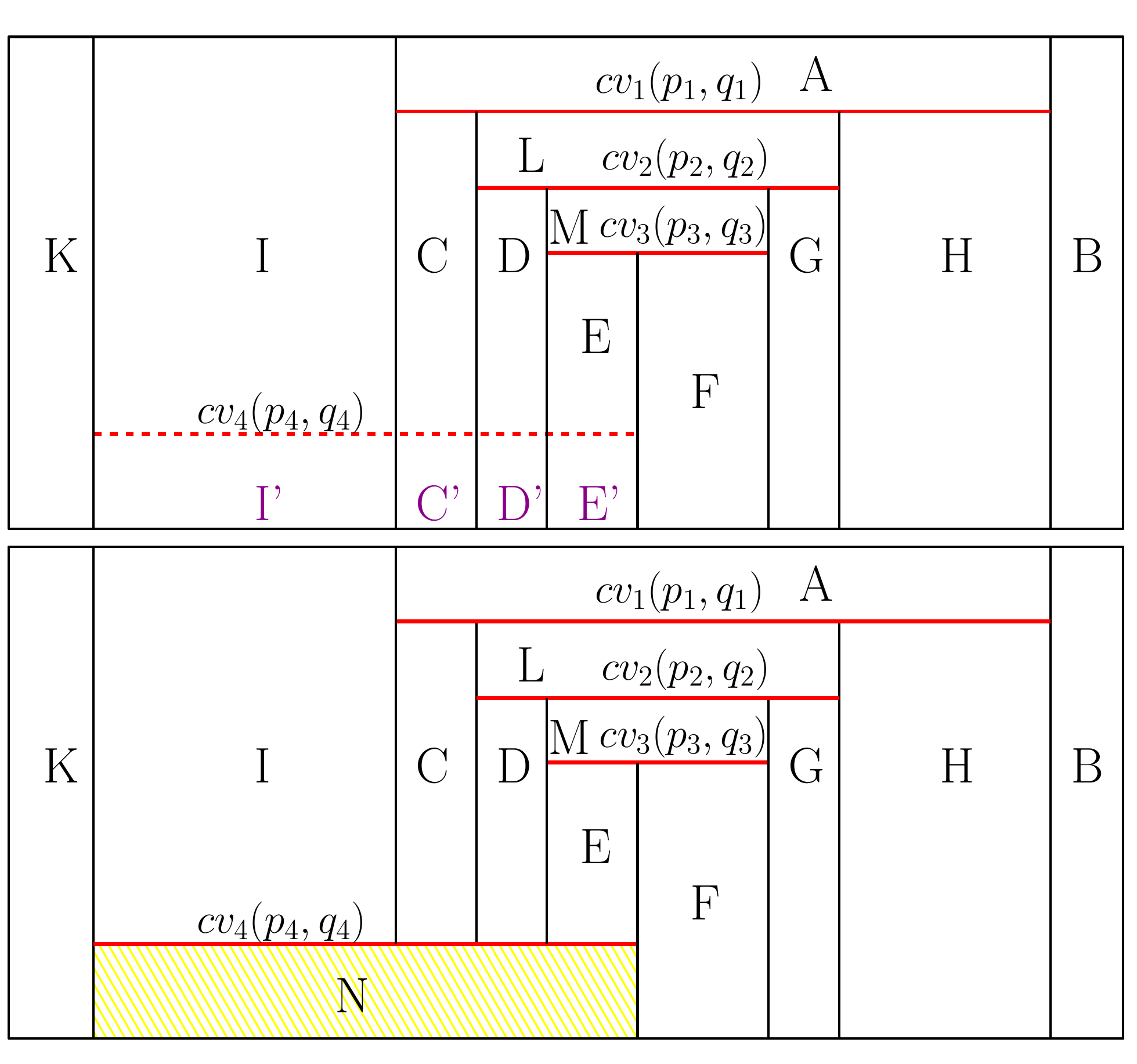}} &
    \subfigure[]{\includegraphics[width=\subfigwidth]{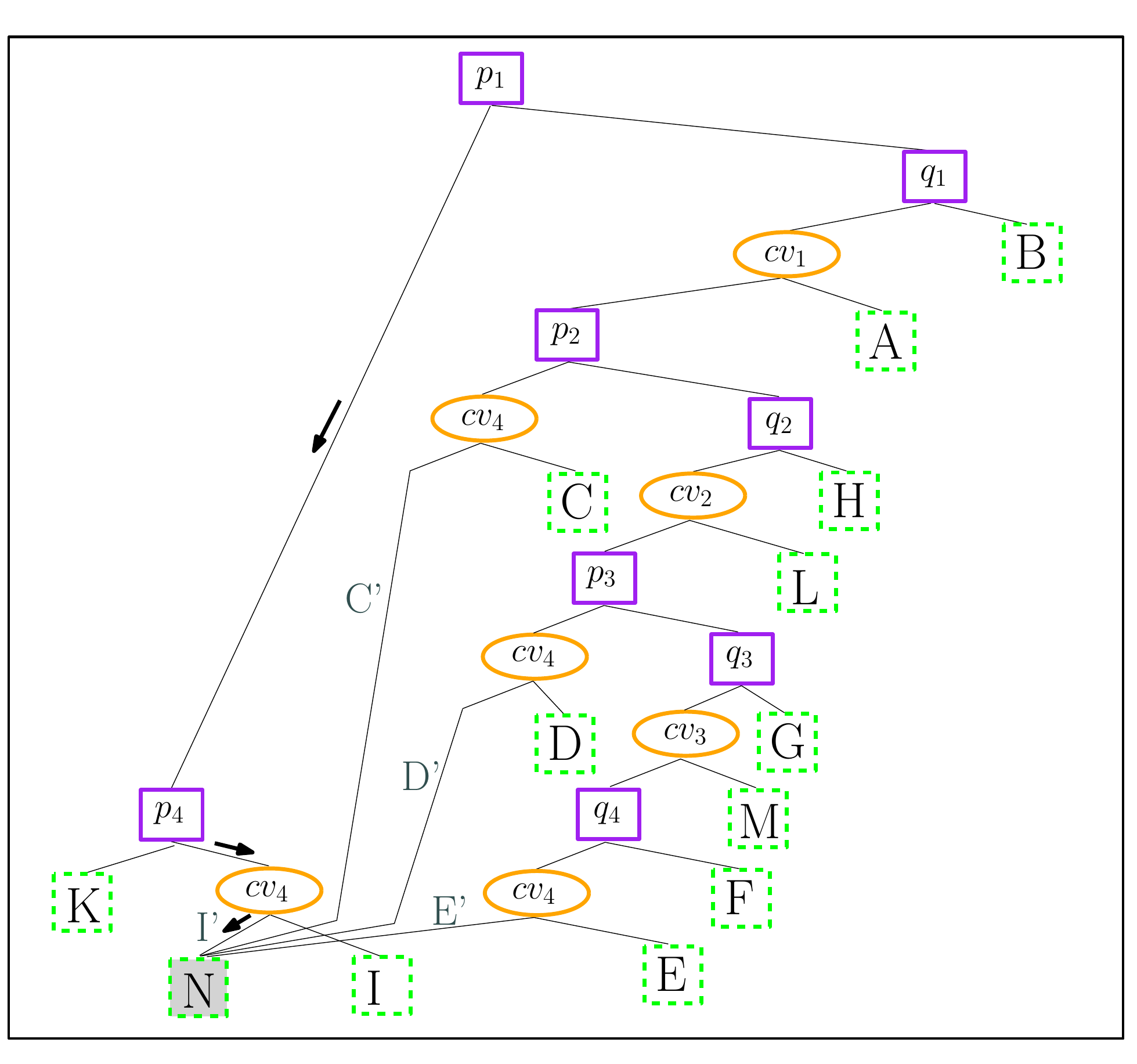}}
  \end{tabular}

  \caption[Demonstrating the $\Omega(\sqrt{n})$ Ratio Construction]
  {\capStyle{\subref{fig:various-a}
      The trapezoidal map after inserting $cv_4$.
      The map is displayed before and after the
      merge of $I'$, $C'$, $D'$, and $E'$ into $N$,
      in the top and bottom illustrations, respectively.
      \subref{fig:various-b} The DAG after merging.
      A query path to the region of $I'$ in $N$ will take 3 steps,
      while the depth of $N$ in this example is 11.
    }}
  \label{fig:depth_vs_qlen:wc_ratio:sqrt_detail}
\end{figure}
}{}

\arxivdcg{
\subsection{Worst Case Ratio} 
\label{subsec:depth_vs_qlen:wc_ratio:wc_construction}}

The following construction,
which uses a recursive scheme,
establishes the worst-case lower bound $\Omega(n/\log{n})$ for~\depth/\lqpl.
\arxivdcg{
Blocks are constructed and arranged in a similar fashion as in the
previous construction. However, this time we have $\log_2{n}$ blocks,
where block $i$ contains $n/2^i$ segments. Within each block
we then apply the same scheme recursively as depicted in
Figure~\ref{fig:depth_vs_qlen:wc_ratio:wc_construction}.
Again segments are inserted from top to bottom such that
the depth of $\Omega(n)$ is achieved in the trapezoid below the
lowest segment.
}
{
There are $\log_2{n}$ blocks,
where block $i$ contains $n/2^i$ segments.
Within each block
the same scheme is applied recursively, as depicted in
Figure~\ref{fig:depth_vs_qlen:wc_ratio:wc_construction}.
The segments are inserted from top to bottom such that
the depth of $\Omega(n)$ is achieved in the trapezoid below the
lowest segment.
}
The fact that the lengths of all search paths
\begin{figure}[t] 
 \centering
  \includegraphics[width=0.6\textwidth]{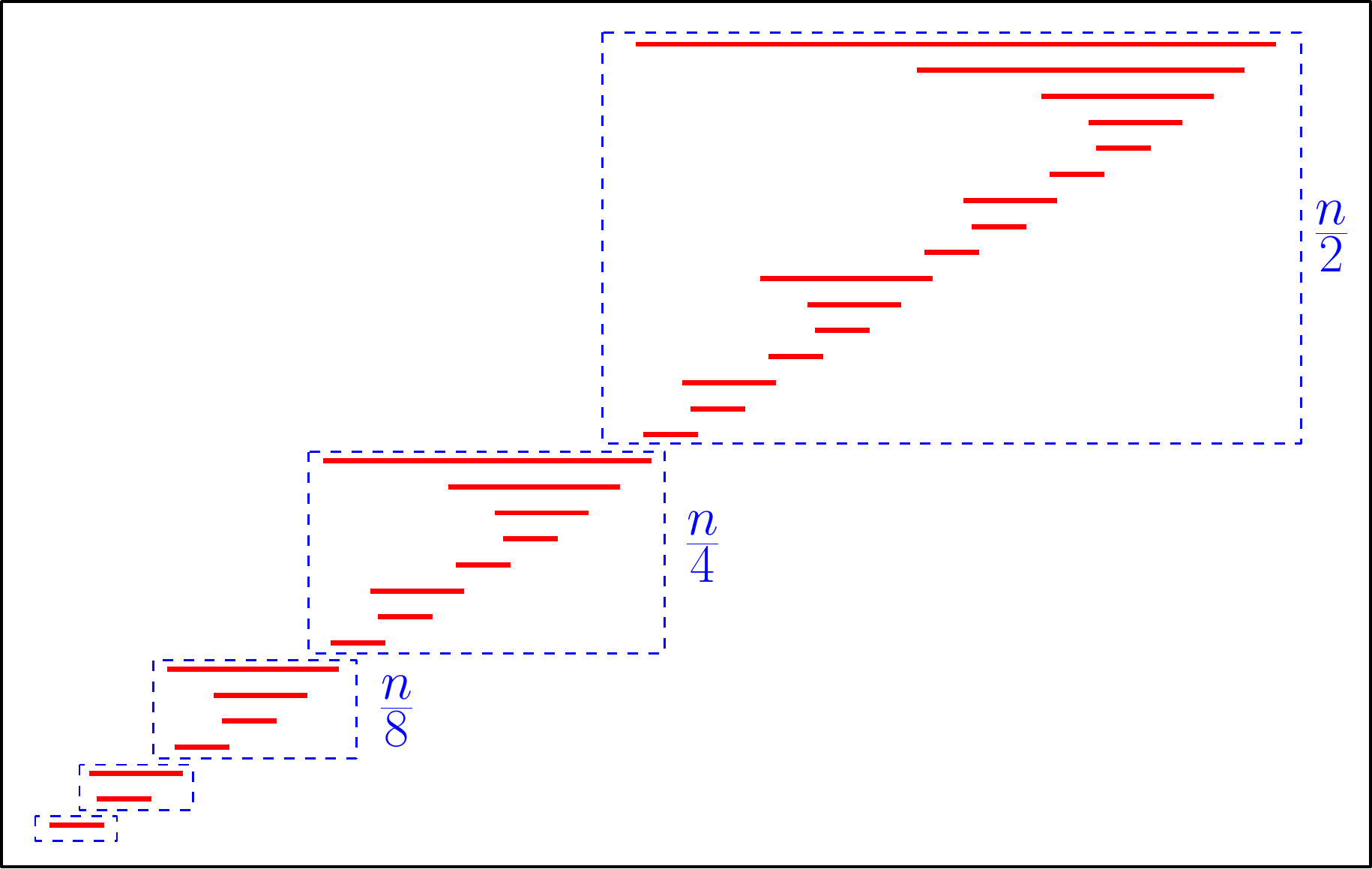}
   \caption[Demonstrating the $\Omega(n/\log{n})$ Ratio Construction]
  {\capStyle{
      A recursive construction establishing
      the $\Omega(n/\log{n})$ lower bound for the~\depth/\lqpl ratio.
    }}
  \label{fig:depth_vs_qlen:wc_ratio:wc_construction}
\end{figure}
are logarithmic
can be proven by the following argument.
By induction we assume that the longest search path within a block of
size $n/2^i$ is some constant times $(\log_2{n} - i)$.
Obviously this is true for a block containing only one segment.
Now, in order to reach block~$i$ containing~$n/2^i$ segments,
we require~$i-1$ comparisons to skip the~$i-1st$ preceding blocks.
Thus in total the search path is of logarithmic length.

\begin{theorem}
The $\Omega(n/\log{n})$ worst-case lower bound on~\depth/\lqpl is tight.
\end{theorem}
\begin{proof}
Obviously,~\depth of $O(n)$ is the maximal achievable depth, since
by construction each segment can only appear once
along {\em any} path in the DAG.
It remains to show that for any scenario with $n$ segments
there is no DAG for which~\lqpl
is smaller than $\Omega(\log{n})$.
Since there are $n$ segments, there are at least~$n$ different
trapezoids having these segments as their top boundary.
Let~$T$ be a decision tree of the optimal search structure
in the sense that its longest query path is the shortest possible.
Each path in the decision tree corresponds to a valid search path
in the DAG and vice versa.
The depth of~$T$ must be at least~$\log_2{n}$, since it is only
a binary tree.
We conclude that the worst-case ratio $\depth/\lqpl$ is
$\Theta(n/\log{n})$.
\myqed
\end{proof}

 \section{A Bijection between the Search Paths in the History DAG and in the Trapezoidal Search Tree}
\label{sec:relation_G_and_T}

Let $S$ be a set of $n$
pairwise interior disjoint $x$-monotone curves
inducing a planar subdivision.
The trapezoidal search tree~\TrpSrchTree
for $S$ is a full binary tree constructed as the DAG~\G using the same insertion order
while skipping the merge step.

The DAG~\G
has an expected linear size~\cite{Mulmuley1990_fast_planar_part_alg,Seidel1991_simp_fast_inc_rand_td}.
On the other hand, the trapezoidal search tree~\TrpSrchTree requires
$\Omega(n\log{n})$ memory for certain scenarios, as shown
in~\cite{SeidelA2000_wc_query_comp}.
The following lemma,
which seems to be folklore,
bounds the expected size of \TrpSrchTree.
For completeness, we give a proof in Appendix~\ref{appendix_trapezoidal_search_tree}.

\begin{lemma}
\label{lemma:tree_size}
Let $S$ be a set of $n$
pairwise interior disjoint
$x$-monotone curves
inducing a planar subdivision.
The expected number of leaves
in the trapezoidal search tree \TrpSrchTree,
which is constructed as the DAG but without merges,
is $O(n\log{n})$.
\end{lemma}

As we show next,
there is a bijection between all possible search paths in~\TrpSrchTree and those of~\G,
even though the two structures differ in size.
First, let us define here the notion of {\it  bouncing nodes}.
Suppose we query~\G with point~$q$.
Additionally, assume that while searching for~$q$ we maintain an interval of the $x$-values that
are still possible with respect to the decisions taken so far.
This {\em history interval} is updated at each decision node according to the
following scheme:
(i) if the node is a curve node then the history interval does not change
(ii) if the node is a point node whose $x$-coordinate is contained in the
current interval, then the interval is updated according to the position of $q$
(iii) if the node is a point node whose $x$-coordinate is not contained in the current interval
then the interval remains unchanged.
Such a point node that is not contained in the current history interval of the path is
named a {\bf bouncing-node for the corresponding path} in \G;
Figure~\ref{fig:relation:bouncing} gives an example of a bouncing node.

\begin{figure} [h] 
  \centering
  \includegraphics[width=0.7\textwidth]{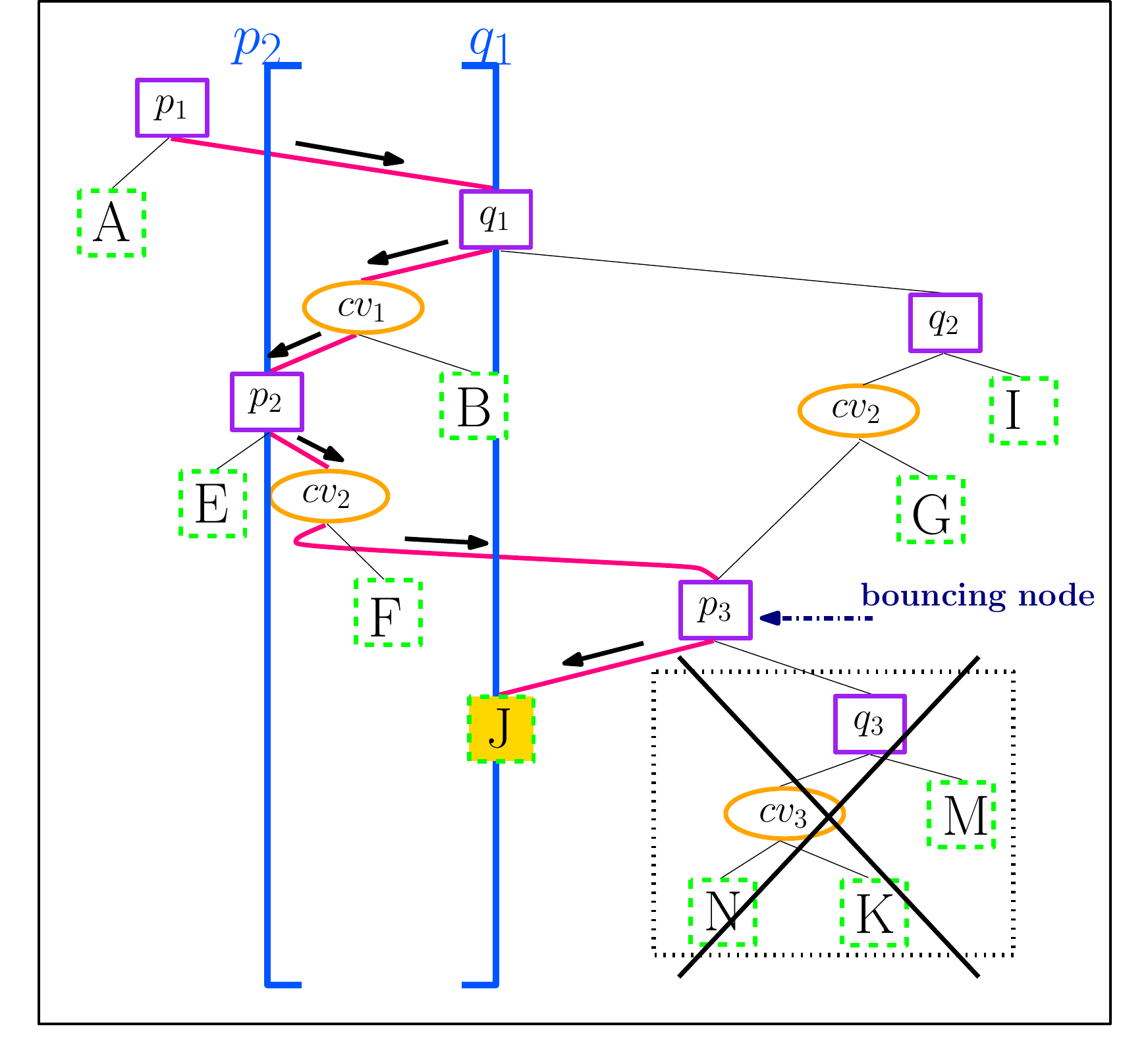}
  \caption{
    \capStyle{%
      The node $p_3$ is a {\bf bouncing node} for the path leading
      into the left part of trapezoid~$J$
      (see also Figure~\ref{fig:depth_vs_qlen:wc_ratio:third_curve}).
      The history interval when reaching node $p_3$ is $(p_2,q_1)$.
      The decision at $p_3$ is already predetermined by the history of the path:
      $q_1$ is to the left of $p_3$, the fact that the path descended to the left at
      node $q_1$ implies that it also descends to the left at
      $p_3$.
    }}
  \label{fig:relation:bouncing}
\end{figure}

The following proposition shows that each search path in~\G has a
corresponding path in~\TrpSrchTree (and vice versa),
and these two paths are identical up to additional bouncing nodes in the path in~\G.

\begin{proposition}
\label{prop:identical_search_paths_per_query}
Let $S$ be a set of $n$
pairwise interior disjoint $x$-monotone curves
inducing a planar subdivision.
Let \G and~\TrpSrchTree be the DAG and the trapezoidal search tree
created using the same permutation of the curves in $S$, respectively.
There exists a canonical bijection among all search paths in~\G and those of~$\TrpSrchTree$,
that is, for any query point $q$, the corresponding search paths for $q$
in~\G and~\TrpSrchTree are identical up to bouncing nodes.
\end{proposition}
\begin{proof}

Let~$q$ be a query point and~$t$ and~$t'$ be the leaf trapezoids containing~$q$ 
in~$\G$ and~\TrpSrchTree, respectively. Obviously, the top and bottom curves
of~$t$ and~$t'$ are identical and~$t$ covers~$t'$ since merges are only allowed
in~\G.
Suppose that while searching for~$q$ we maintain the history-interval of possible~$x$-values.
At the end of the search in~\TrpSrchTree this interval is obviously identical
to the~$x$-range of~$t'$.
We show by induction the bijection between the two search paths for~$q$ and in
fact we also show that the history intervals maintained while searching in~\G
and~\TrpSrchTree are identical, i.e., the history interval that is eventually
obtained by the search in~\G is identical to the~$x$-interval of~$t'$.

Let~$\G_{i}$, $\TrpSrchTree_{i}$ denote the DAG and the trapezoidal search tree
after the first~$i$ curves were inserted, respectively.
We denote by~$t_{i}$ and~$t_{i}'$ the trapezoids containing the query point~$q$
in~$\G_{i}$ and~$\TrpSrchTree_{i}$, respectively.
Let~$(a_{i}, b_{i})$ and~$(a'_{i}, b'_{i})$ denote the~$x$-intervals 
of~$t_i$ in~$\G_i$ and~$t'_i$ in~$\TrpSrchTree_i$, respectively.
The base case is trivial since~$\G_1 = \TrpSrchTree_1$.
Now suppose that the statement holds for~$i-1$.
We show that it holds for~$i$ as well.
The~$i$th curve~$cv_i(p_i,q_i)$ is now inserted into both~$\G_{i-1}$ and~$\TrpSrchTree_{i-1}$.
The basic argument is as follows. For both endpoints of~$cv_i$ there are essentially
three cases:
(i) the point is outside $t_{i-1}$ and $t'_{i-1}$: the point has no effect on both paths.
(ii) the point is inside $t_{i-1}$ and $t'_{i-1}$: the point shows up as a normal node in both paths
and the (identical) history intervals are updated accordingly.
(iii) the point is inside $t_{i-1}$ but not in $t'_{i-1}$: the point has no effect on the search path
in $T_i$ while it may show up on the search path in~$\G_i$, but only as a bouncing node, i.e., the
history interval remains unchanged.
Figure~\ref{fig:induction_all} shows the 15 possible positions to insert $cv_i$ with respect to
$t_{i-1}$ and $t'_{i-1}$.

\begin{figure}[h]
    \centering
    \includegraphics[width=0.6\textwidth]{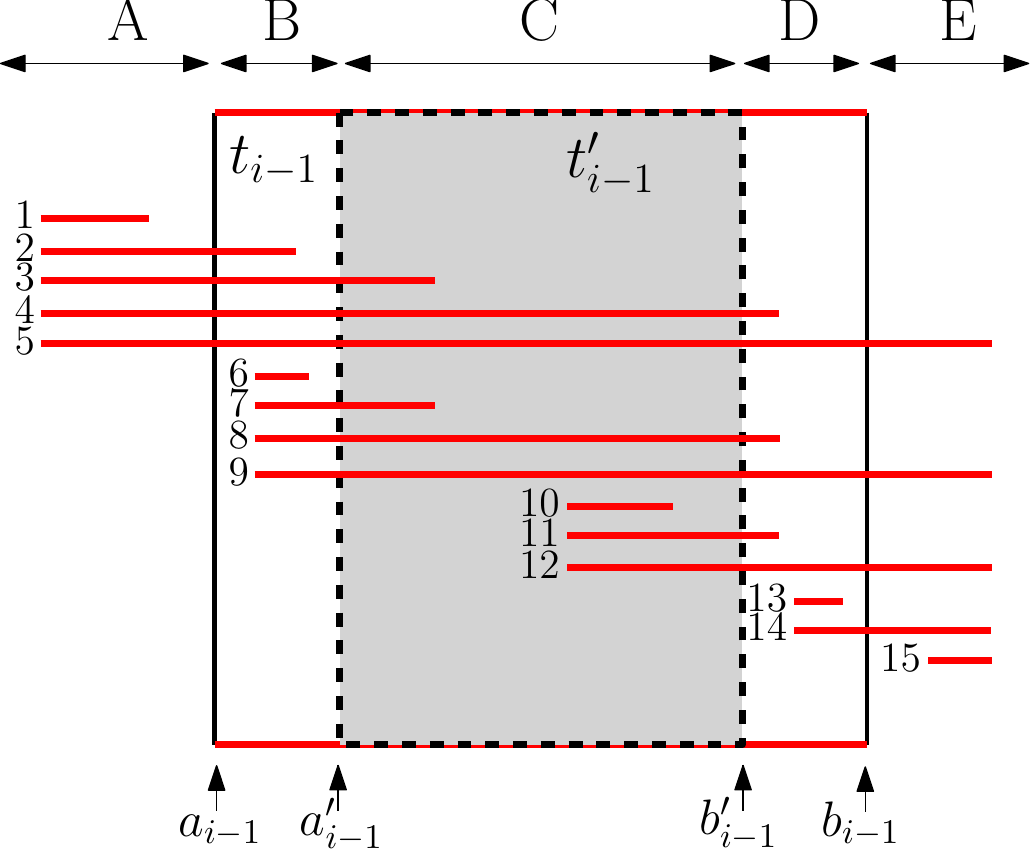}
     \caption[Optional Next-Step in the Induction]{
       \capStyle{Possible positions for $cv_i$ in relation to trapezoid $t_{i-1}$ in $\G_{i-1}$, which covers trapezoid $t'_{i-1}$ in $\TrpSrchTree_{i-1}$. }
       \label{fig:induction_all}}
\end{figure}

As an example we discuss position 13, while the full and
rather straightforward case analysis
is given in Appendix~\ref{appendix_bijection}.
In this case $p_i$ as well as $q_i$
are inside $t_{i-1}$ but to the right of $t'_{i-1}$.
The search path for $q$ in $T_i$ remains unchanged since $t'_{i-1}$ is not destroyed
whereas the path in $\G_i$ changes as $t_{i-1}$ is destroyed.
However, the only change is the addition of $p_i$, which is a bouncing node for that path
since it is not contained in the history-interval, i.e.,
the $x$-range remains unchanged since $t'_{i-1} = t'_i$.
Notice that, in this particular case, the right end point $q_i$ does not even appear as
a bouncing node since it is shadowed by $p_i$.

\myqed
\end{proof}


\begin{lemma} \label{lemma:charging_edges}
Every edge $e'\in \TrpSrchTree$ can be associated to precisely one
sequence of edges in the corresponding DAG \G.
\end{lemma}
\begin{proof}
Since \TrpSrchTree is a tree, all search paths in \TrpSrchTree that use $e'$
are identical up to that point.
Thus, the decisions taken at intermediate bouncing nodes
while following the corresponding path in \G are predetermined due
to their common history.
\myqed
\end{proof}

\noindent Hence, in the following we say that $e'\in \TrpSrchTree$ accumulates
bouncing nodes, namely the bouncing nodes on its corresponding
subpath in $\G$.


\begin{observation} \label{obs:bouncing_nodes_on_leaf_edges}
Let $\G_{i-1}$, $\TrpSrchTree_{i-1}$ be as defined above.
Only edges in $\TrpSrchTree_{i-1}$ that currently end in
leaves may accumulate additional bouncing nodes due to
the insertion of the $i$th curve.
More precisely, let $t'_{i-1}$ be the trapezoid in which the
leaf edge $e'$ ends and let $t_{i-1}$ be the trapezoid in $\G_{i-1}$
that covers $t'_{i-1}$, as illustrated in Figure~\ref{fig:induction_all}.
The edge $e'$ may only accumulate additional bouncing nodes iff
$t_{i-1}$ is destroyed.
\end{observation}

\begin{definition}
\label{def:critical_edge}
An edge of \TrpSrchTree that at some intermediate step of the construction
was a leaf edge, that is, ended in a trapezoid,
is named a {\em critical edge}.
\end{definition}

\noindent
Note the direct correspondence between a critical edge and its trapezoid.
An edge remains a leaf edge until its trapezoid is destroyed,
in which case the trapezoid is replaced by an internal node
of \TrpSrchTree.


\begin{observation} \label{obs:max_bounce_per_curve}
The insertion of a single curve may incur at most two additional
bouncing nodes for a
search path in $\G_i$.
\end{observation}

 \section{Efficient Construction Algorithms for Static Settings}
\label{sec:eff_ver_alg}

Given a set $S$ of $n$
pairwise interior disjoint $x$-monotone curves
inducing a planar subdivision, we seek to devise
an efficient construction algorithm for static settings,
when all input curves are given in advance,
which results in a linear-size DAG and
logarithmic query time in the worst case.
Theorem~\ref{thrm:gnrl_preproc_time} gives resource bounds on such
an algorithm based on the availability of an efficient verification
algorithm for \lqpl.

\begin{definition}
Let $f(n)$ denote the time it
takes to verify that, in a linear size DAG constructed over
 a set of~$n$
pairwise interior disjoint $x$-monotone curves,~\lqpl is bounded by $c \log n$
for a constant~$c$.
\end{definition}

\begin{theorem}
\label{thrm:gnrl_preproc_time}
Let $S$ be a set of $n$
pairwise interior disjoint $x$-monotone curves
inducing a planar subdivision.
A \PL data structure for $S$, which has~$O(n)$
size and~$O(\log{n})$ query time in the worst case,
can be built in $O(n\log{n} + f(n))$ expected time, where
$f(n)$ is as defined above.
\end{theorem}

\begin{proof}
The construction of a DAG with some random insertion order takes expected
$O(n\log n)$ time. The linear size can be verified trivially on the fly
(as discussed in Subsection~\ref{subsec:background:guaranteed}).
After the construction an algorithm, requiring~$f(n)$ time, that verifies
that the maximum query path length~\lqpl is logarithmic is used.
The verification of the size~\calS and
the maximum query path length~\lqpl may trigger a rebuild with a new
random insertion order.
However, according to Lemma~\ref{lemma:prelim:guaranteed:constant_rebuilds}
one can expect
only a constant number of rebuilds.
Thus, the overall expected running time remains~$O(n\log{n} + f(n))$.
\myqed
\end{proof}

The next two subsections describe
two efficient verification
algorithms for~\lqpl that can be used by the
general construction algorithm.
The first one uses the existing search structure and
has expected~$O(n\log{n})$ running time.
The second algorithm is less straightforward to apply
but has worst-case~$O(n\log n)$ running time.

\subsection{An Expected $O(n\log {n})$ Verification Algorithm}
\label{subsec:eff_ver_alg:rec}

The following algorithm verifies that the
maximum query path length~\lqpl in the search structure
is bounded by $c_{\lqpl}\log(n)$, where $c_{\lqpl}$ is some properly chosen constant
according to~\cite[Sec~6.4]{CG-alg-app}.
%
The algorithm is recursive, starting at the root it descends
towards the leaves and explores all possible search paths and
discards all other paths that are geometrically unrealizable.
To do so, each recursion call maintains the history-interval of $x$-values that
are still possible with respect to the decisions taken so far.

The algorithm starts at the root with the maximal interval,
i.e., $[-\infty,+\infty]$.
At each node there are three possible cases:
(i) the recursion reaches a curve node, it splits for the upper and
lower path while the interval remains unchanged;
(ii) the node is a point node whose $x$-coordinate is contained in the
current interval $I$, the recursion splits to the left and the right side
with updated intervals, i.e., $I$ is split at the $x$-coordinate of the node;
(iii) the node is a point node whose $x$-coordinate is {\it not} contained in $I$
(bouncing node for this path),
the recursion does not split and continues to the proper child only with $I$
unchanged.
Figure~\ref{fig:eff_ver_alg:rec_alg} illustrates a partial run of the algorithm.

\begin{figure}[h]
  \centering
  \setlength{\spacewidth}{0pt}
  \setlength{\subfigwidth}{0.34\textwidth}
    \hspace{-2mm}\subfigure[]{\includegraphics[width=\subfigwidth]{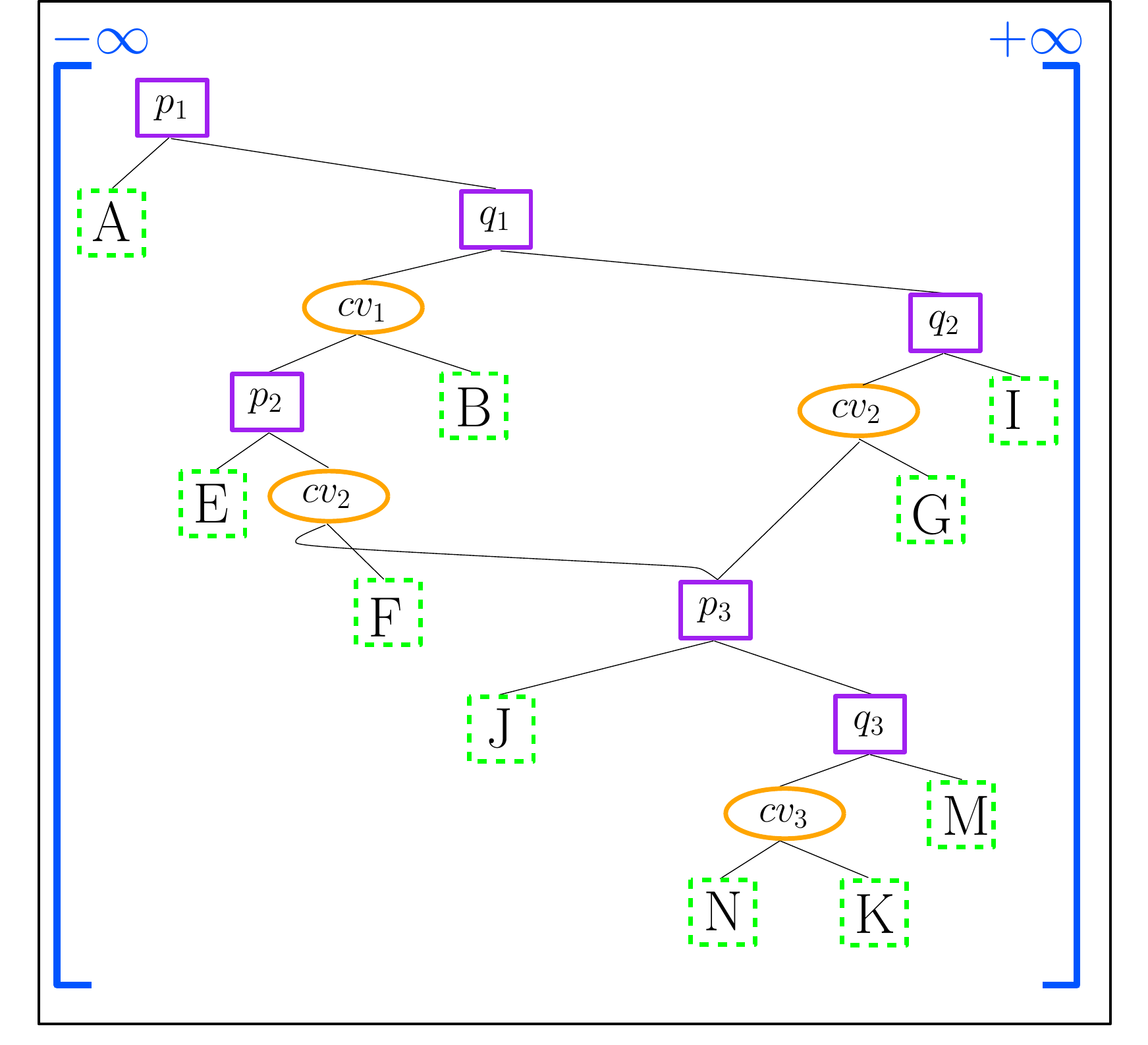}} \hspace{-1.5mm}
    \subfigure[]{\includegraphics[width=\subfigwidth]{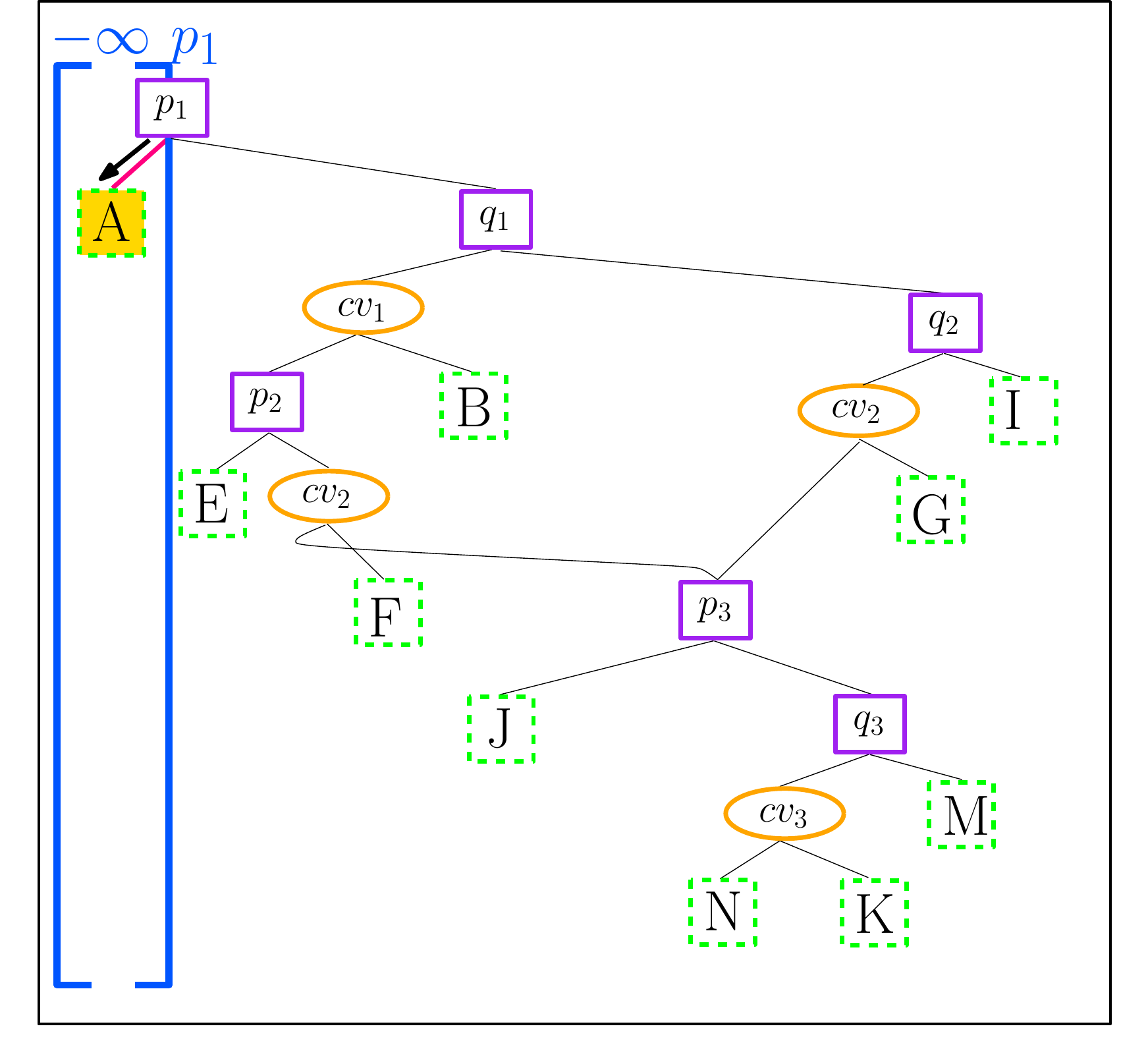}} \hspace{-1.5mm}
    \subfigure[]{\includegraphics[width=\subfigwidth]{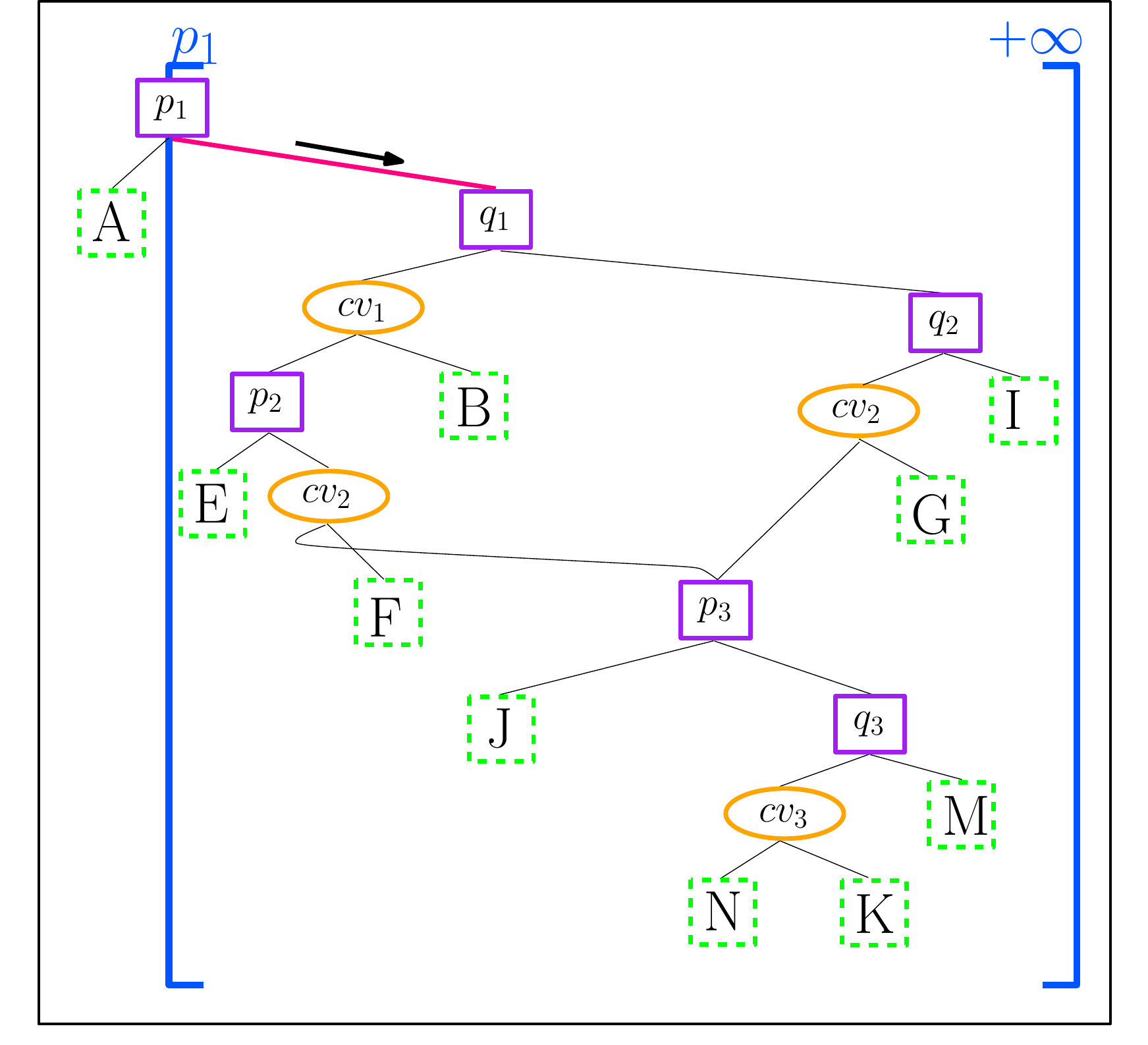}}\hspace{-3mm}\\
    \hspace{-2mm}\subfigure[]{\includegraphics[width=\subfigwidth]{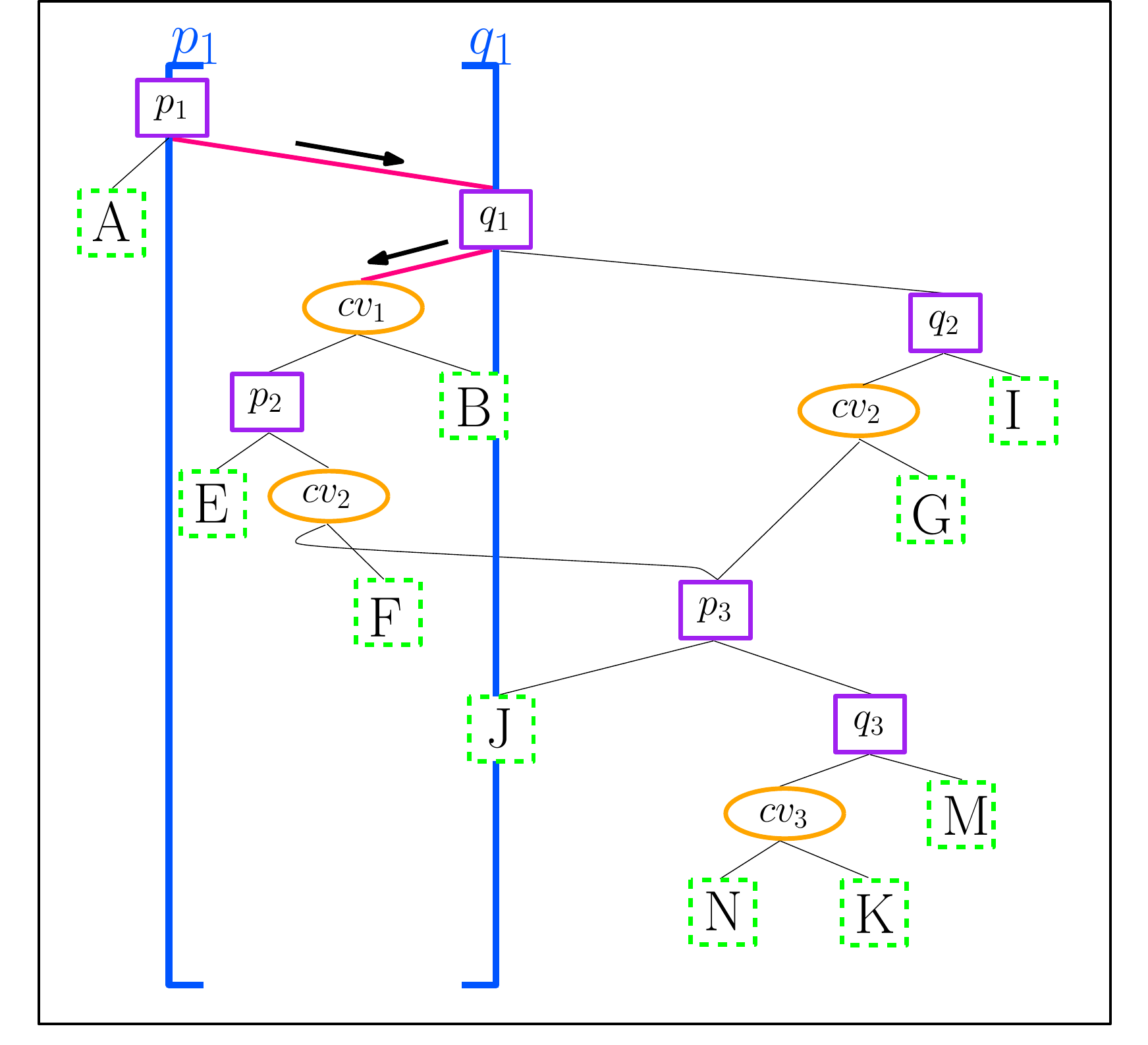}} \hspace{-1.5mm}
    \subfigure[]{\includegraphics[width=\subfigwidth]{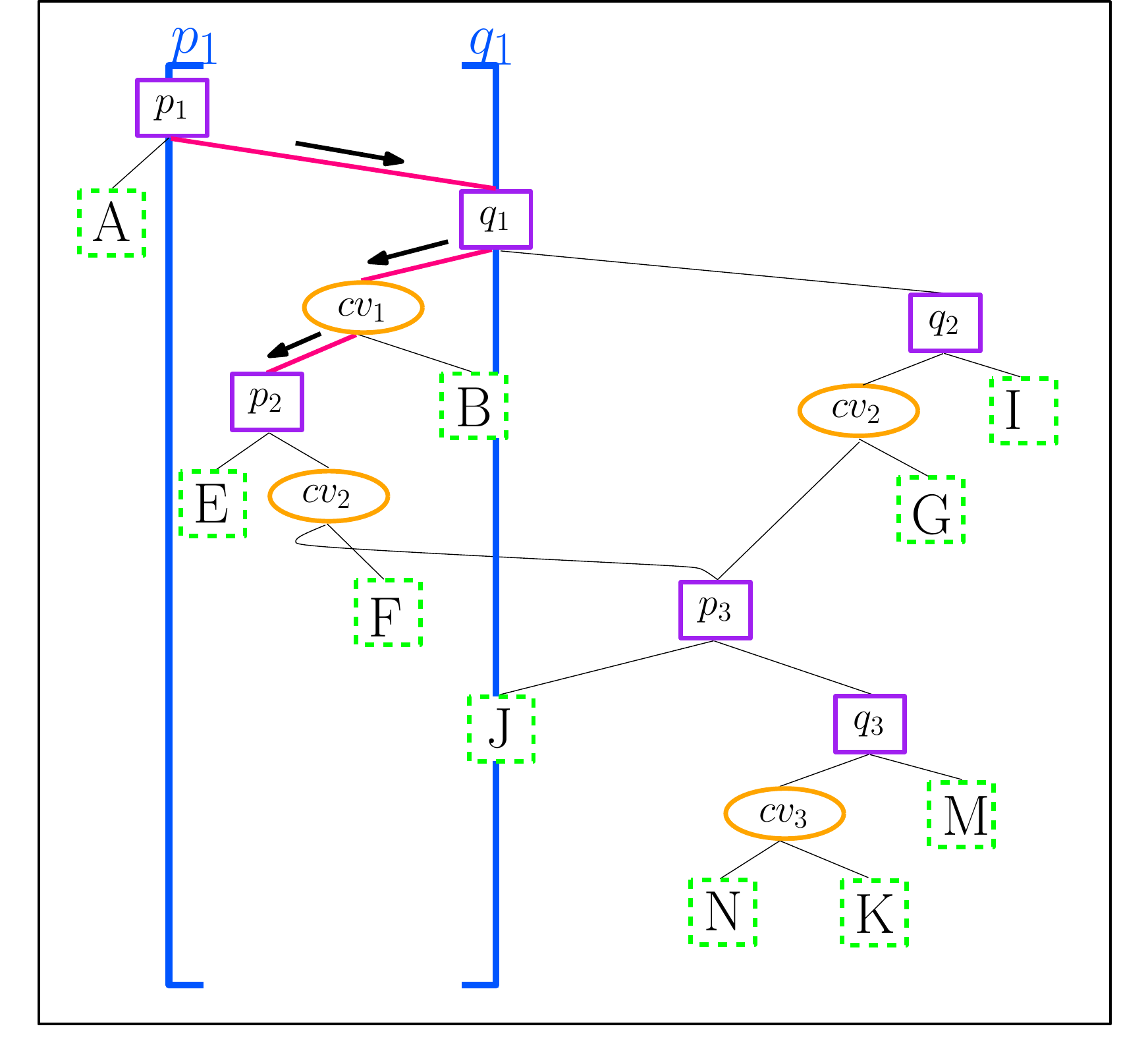}} \hspace{-1.5mm}
    \subfigure[]{\includegraphics[width=\subfigwidth]{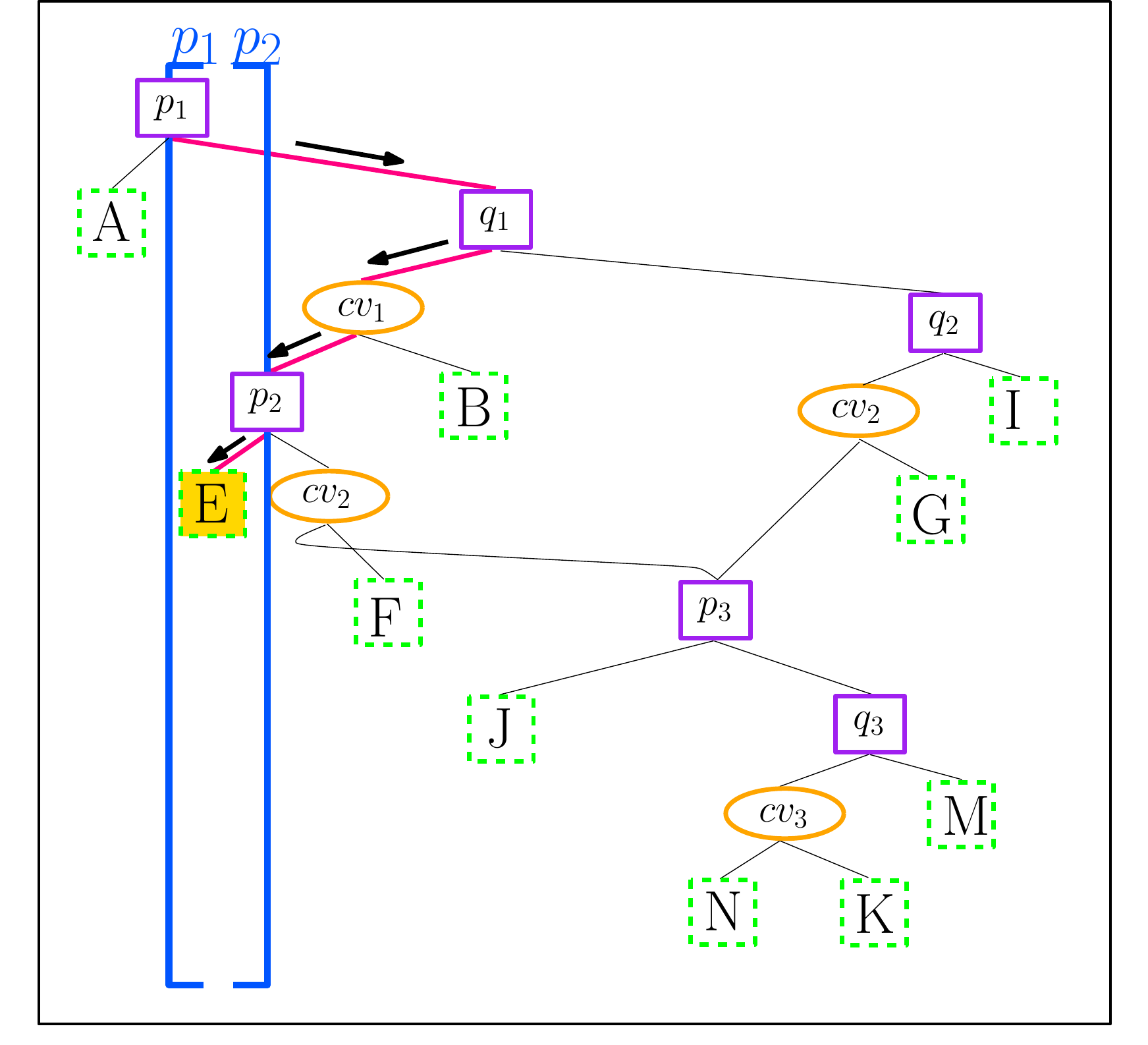}}\hspace{-3mm}\\
    \hspace{-4mm}\subfigure[]{\includegraphics[width=\subfigwidth]{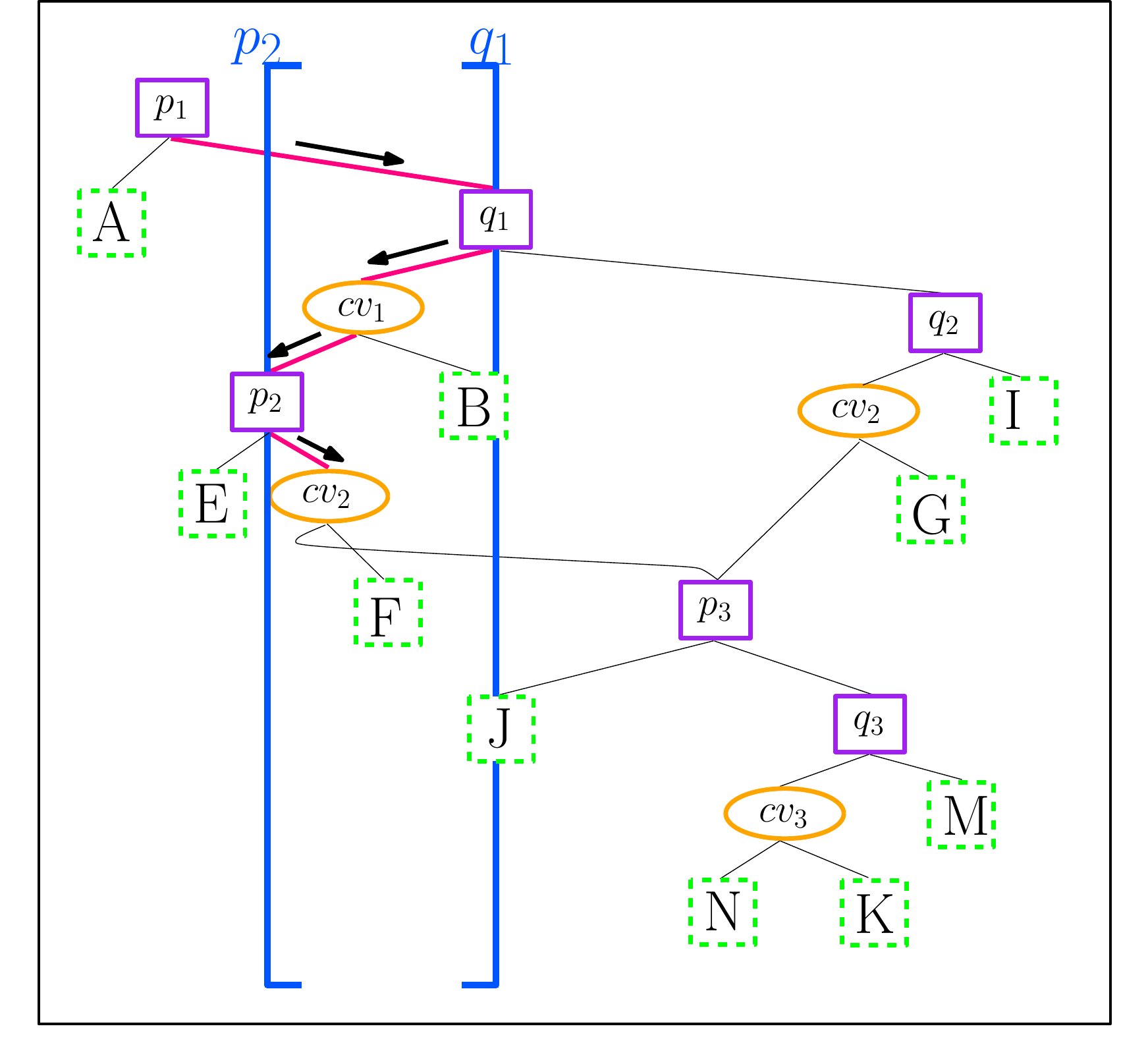}} \hspace{-1.5mm}
    \subfigure[]{\includegraphics[width=\subfigwidth]{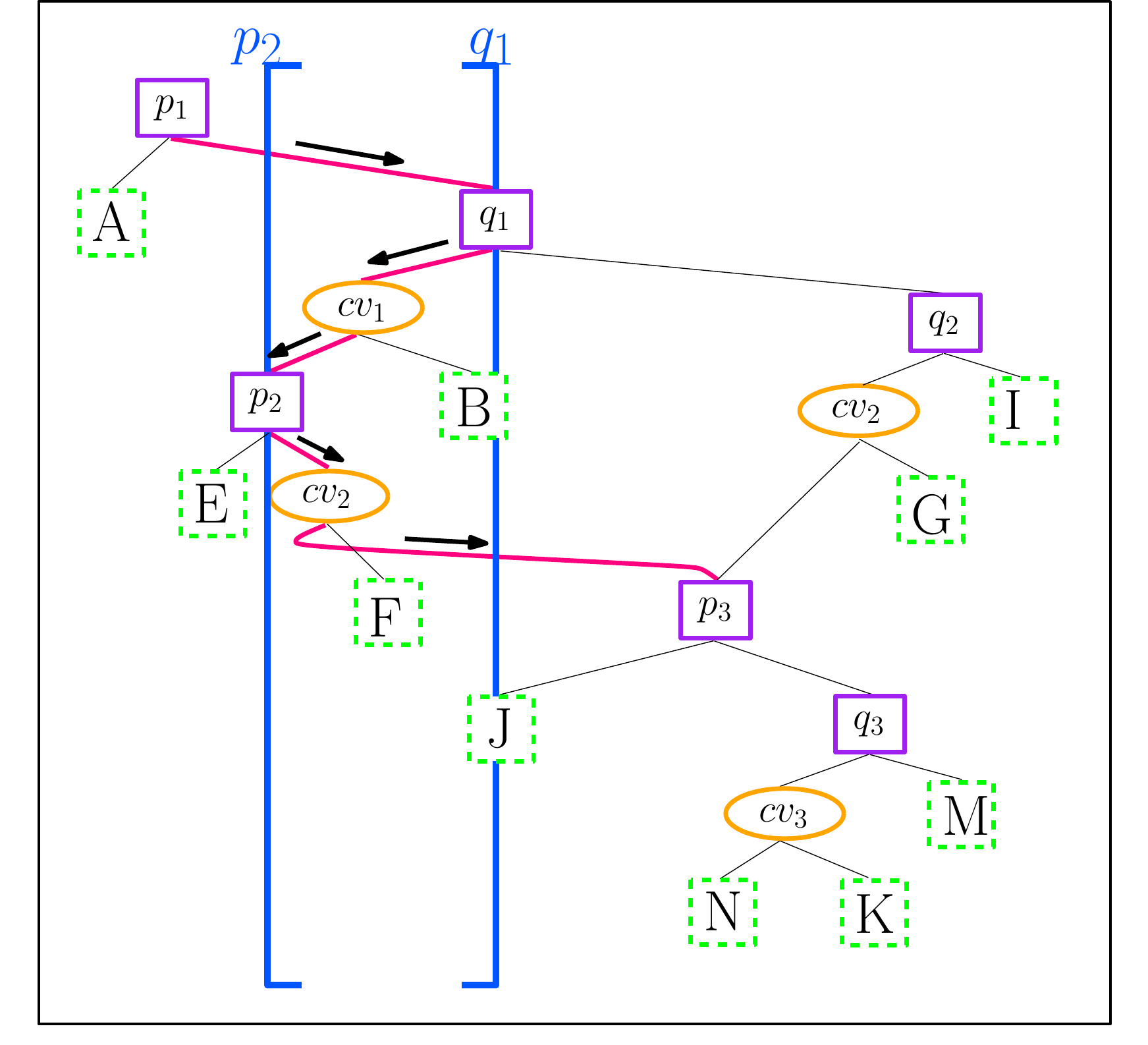}} \hspace{-1.5mm}
    \subfigure[]{\label{fig:rec-i}\includegraphics[width=\subfigwidth]{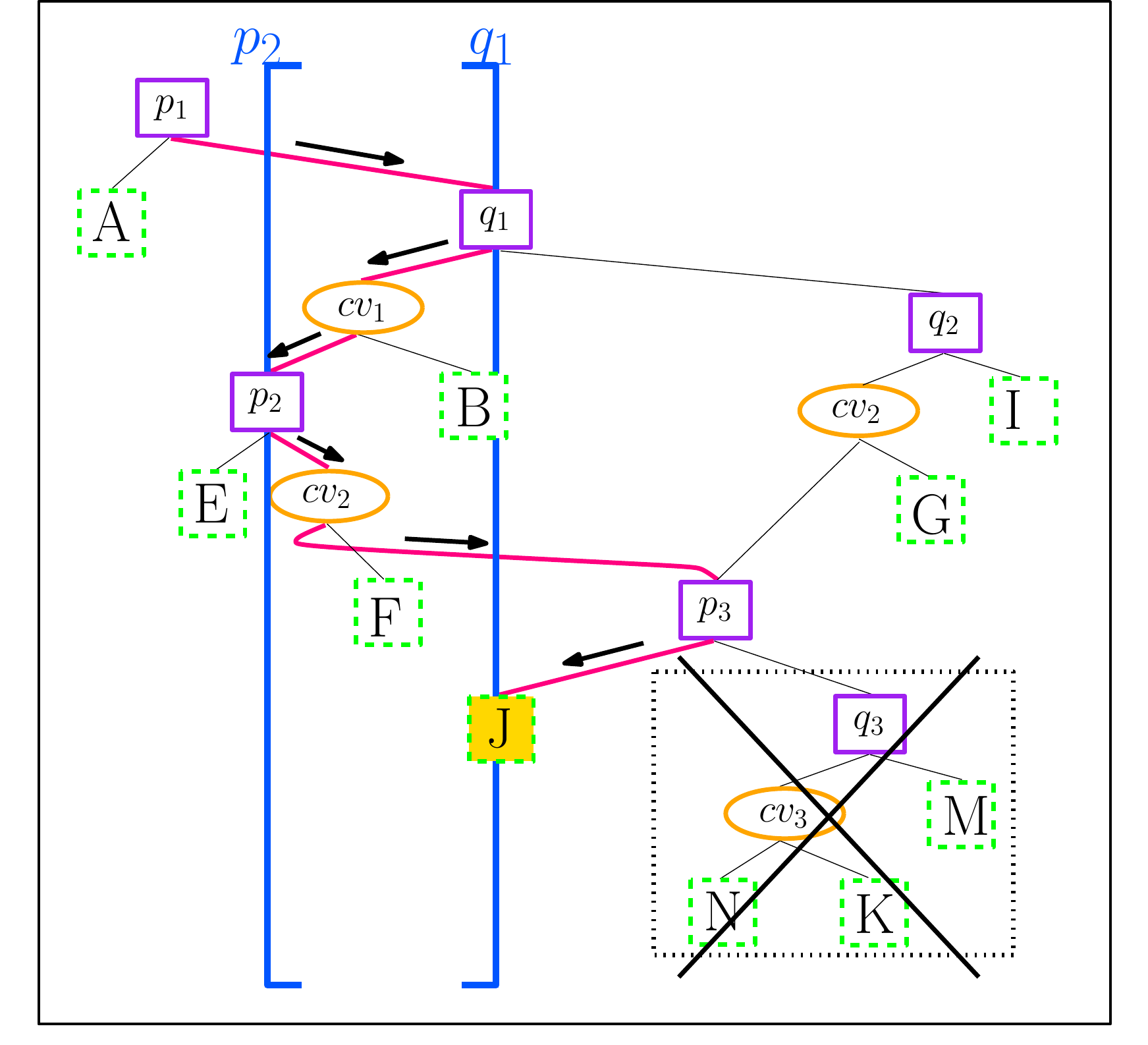}}\hspace{-3mm}

  \vspace{-5pt}
   \caption[Recursive Verification Algorithm Flow]{
   \capStyle{The first 9 steps of the recursive verification algorithm
   run on the search structure for 3 curves,
   as illustrated in Figure~\ref{fig:depth_vs_qlen:wc_ratio:third_curve}.
   The interval of possible $x$-values is marked by the blue brackets.
   In each step the growing path so far is marked with arrows.
   In \subref{fig:rec-i} the interval of possible $x$-values
   remains~$[p_2,q_1]$ and does not shrink
   since $p_3$ is not contained in it.
   The subgraph rooted at the right child of $p_3$
   is clearly not contained in $[p_2,q_1]$,
   since it represents regions that are
   completely to the right of $p_3$, and
   is, therefore, skipped. $p_3$ is a bouncing node for
   the path depicted in (i).
   \vspace{-7pt}
   }}
  \label{fig:eff_ver_alg:rec_alg}
\end{figure}

The expected running time of the above recursive algorithm applied
to \TrpSrchTree would be $O(n\log n)$.
This follows from the fact that the algorithm would use each edge of the tree
exactly once and by the expected size of the tree,
which by Lemma~\ref{lemma:tree_size} is $O(n\log n)$.
In fact, the behavior of the algorithm when applied to the corresponding DAG \G
is very similar since the bouncing nodes create additional
costs but do not let the recursion split.
That is, the algorithm still follows each edge $e'$ of \TrpSrchTree
but with extra costs per edge incurred by bouncing nodes,
see also Lemma~\ref{lemma:charging_edges}.
Thus, the total cost of the recursive verification algorithm
applied to~\G~is

\begin{equation}
\label{eq:recursive_0}
\sum_{e'\in\TrpSrchTree}(w_{e'}+1),
\end{equation}
\noindent
where $w_{e'}$ is the number of additional bouncing nodes in the corresponding
subpath for $e'$ in \G. 
By Observation~\ref{obs:bouncing_nodes_on_leaf_edges} we know that
only critical edges may accumulate many bouncing nodes.
On the other hand, by Observation~\ref{obs:max_bounce_per_curve},
all other edges can only accumulate up to two bouncing nodes.
Each of the critical edges can be associated with a trapezoid
that existed during the construction of \TrpSrchTree.
Hence, let $\Delta^\TrpSrchTree$ denote the set of all trapezoids that
were created during the construction of $\TrpSrchTree$.
For every such trapezoid $t'\in\Delta^\TrpSrchTree$ we define its weight
as $w_{t'} = w_{e'}$, where $e'$ is its corresponding
critical edge in $\TrpSrchTree$.
For all other edges the total cost is at most 3
(at most 2 bouncing nodes by Observation~\ref{obs:max_bounce_per_curve}).
Hence, we can upper bound~(\ref{eq:recursive_0}) as follows:

\begin{equation}
\label{eq:recursive_1}
\sum_{e'\in\TrpSrchTree}(w_{e'}+1) \leq 3|\TrpSrchTree| + \sum_{t'\in\Delta^{\TrpSrchTree}}w_{t'}.
\end{equation}

\noindent
Let $\Delta$ be the set of all possible trapezoids that may exist
during the construction of a trapezoidal search tree.
Now, set $w_{t'}=0$ for all $t'\in\Delta\setminus\Delta^{\TrpSrchTree}$,
\ie for those that were not created with respect to a specific insertion order.
We can now extend the right hand side of (\ref{eq:recursive_1}) without
changing its value as follows:
\begin{equation}
\label{eq:recursive_2}
3|\TrpSrchTree| + \sum_{t'\in\Delta^{\TrpSrchTree}}w_{t'} = 3 |\TrpSrchTree| + \sum_{t'\in\Delta} w_{t'}.
\end{equation}
We are interested in the expected value of the right hand side
of~(\ref{eq:recursive_2}), \ie the expected value with respect to
all~$n!$ possible insertion orders of the segments in $S$.
This can be can be written as
\begin{equation}
\label{eq:recursive_3}
O(n\log n) + \E[\sum_{t'\in\Delta}w_{t'}],
\end{equation}
since, by Lemma~\ref{lemma:tree_size}, the expected size of $\TrpSrchTree$ is $O(n\log n)$.
By linearity of expectation~(\ref{eq:recursive_3}) is equivalent to:
\begin{equation}
\label{eq:recursive_4}
O(n\log n) + \sum_{t'\in\Delta}\E[w_{t'}]
\end{equation}

\noindent
Let $\delta_{t'} = 1$ if $t'\in\Delta^{\TrpSrchTree}$ and $0$ otherwise.
By the law of iterated expectation we can now split up
the expected value according to the condition whether $t'$ exists during
the construction of \TrpSrchTree or not.

\begin{equation}
\label{eq:recursive_4a}
O(n\log n) + \sum_{t'\in\Delta}
\Big(
\E[w_{t'}|\delta_{t'}=1]\Pr[\delta_{t'}=1]
+
\E[w_{t'}|\delta_{t'}=0]\Pr[\delta_{t'}=0]
\Big)
\end{equation}
Observing $\E[w_{t'}|\delta_{t'}=0]=0$ and $\Pr[\delta_{t'}=1]=\E[\delta_{t'}]$,
we are left with:
\begin{equation}
\label{eq:recursive_4b}
O(n\log n) + \sum_{t'\in\Delta}
\E[w_{t'}|\delta_{t'}=1]\E[\delta_{t'}].
\end{equation}

\noindent
Most of the remainder of the section is dedicated to the fact that
$\E[w_{t'}|\delta_{t'}=1]$ is a constant. It is then straightforward
to conclude that the expected running time is $O(n\log n)$;
see Proposition~\ref{pro:recursive_preproc_time} at the end of this section.


Let $\Pi$ be the set of all $n!$ insertion sequences.
Every $\pi\in\Pi$ defines a construction of a trapezoidal
search tree $\TrpSrchTree(\pi)$ and the corresponding DAG~$\G(\pi)$.
Recall that the difference between the trapezoidal map of \TrpSrchTree and
\G are the merges that occur during the construction of \G.
Hence, a trapezoid $t'\in\Delta^{\TrpSrchTree(\pi)}$ may be covered by
several trapezoids of \G during its existence. We denote the number
of these trapezoids by $n_{t'}(\pi)$.
Obviously, $n_{t'}(\pi) =0$ iff $t' \not \in \Delta^{\TrpSrchTree(\pi)}$.

\begin{lemma}
\label{lem:injection_lemma}
Let $\Pi^{n_{t'} \diamond i} =\{\pi\in\Pi | n_{t'}(\pi)\diamond i\}$
for $\diamond\in\{=,<,>,\leq,\geq\}$.
For any integer $i>0$, the number of insertion sequences where
$n_{t'}(\pi) = i$ is greater or equal to the number of
sequences where $n_{t'}(\pi)$ is larger than $i$, that is:
$$|\Pi^{n_{t'} > i}| \leq |\Pi^{n_{t'} = i}|.$$
\end{lemma}
\begin{proof}

We first define a map $\phi_{t'}^i:\Pi^{n_{t'} > i}\rightarrow\Pi^{n_{t'} = i}$
and then show that it is well-defined and injective.

{\em Definition of $\phi_{t'}^i$:} Since $\pi \in \Pi^{n_{t'} > i}$ there is a sequence
of more than $i$ trapezoids from $\Delta^{G(\pi)}$ that cover $t'$.
Let $t$ be the $i$-th last trapezoid in that sequence.
Let $S(t)$ be the set of at most~$4$ segments that define $t$.
Let $s$ be the segment among those in $S(t)$ that
is inserted last with respect to $\pi$.
Notice that $t'$ must already exists when $s$ is inserted since $n_{t'}(\pi) > i > 0$.
Hence, $s$ cannot be in $S(t')$, otherwise it would
contradict the fact that $t'$ already exists.
Therefore, $s$ can only be $\Left(t)$
or $\Right(t)$, which must extend to the left or right, respectively.
If $s$ equals $\Left(t)$ swap it with $\Top(t)=\Top(t')$,
otherwise with $\Bottom(t)=\Bottom(t')$.
Assuming that $s$ was at position $j$ and that the swapped segment $\overline{s}$
was at position $k<j$, the resulting sequence $\phi_{t'}^i(\pi)$ is~$[s_1,\dots,s_{k-1},s,\dots,\overline{s},s_{j+1},\dots,s_n]$.

{\em  $\phi_{t'}^i$ is well-defined:} We must show that $\phi_{t'}^i(\pi)$ is indeed
in $\Pi^{n_{t'} = i}$.
First observe that $t$ is still constructed at position $j$ since
for $\pi$ and $\phi_{t'}^i(\pi)$ the
set of segments inserted until the $j$-th position (inclusive) is
identical.
Also notice that from this position on $\phi_{t'}^i(\pi)$ and
$\pi$ are identical, which implies that the set of trapezoids in the
trapezoidal map that is constructed from now on, is identical for both
permutations. Specifically, $t$ remains the $i$-th last trapezoid that
covers $t'$.

We still need to argue that $t'$ is now constructed with the insertion
of $\overline{s}$.
Obviously, it cannot be constructed earlier since by definition
$\overline{s}$ is either the top or bottom segment.
The important part is that the vertical walls that define $t'$ are not blocked
due to the new insertion order $\phi_{t'}^i(\pi)$;
see also Figure~\ref{fig:phi}.
The only two segments that changed position are $s$ and $\overline{s}$.
The insertion of $\overline{s}$, which defines the top or bottom sides of $t'$,
could block vertical walls.
However, since its position in $\phi_{t'}^i(\pi)$ is later than its position
in $\pi$ it cannot block a wall that it did not block before.
On the other hand, $s$, which is inserted earlier can only be a
left segment that extends to the left or a right segment that extends to the right.
Hence, it cannot intersect the vertical walls of  $t'$ at all.
We conclude that $t'$ is constructed with the insertion of $\overline{s}$
and that $\phi_{t'}^i(\pi) \in \Pi^{n_{t'} = i}$.

{\em  $\phi_{t'}^i$ is injective:}
By definition of $\phi_{t'}^i(\pi)$ the $i$-th last
trapezoid  that covers~$t'$ is still~$t$.
Therefore, the inverse mapping of $\phi_{t'}^i(\pi)$ can be easily defined by
interchanging the role of $\Left(t)$ with $\Top(t)$ and
$\Right(t)$ with $\Bottom(t)$, respectively. \myqed
\end{proof}

\begin{figure}
  \centering
  \includegraphics[width=0.8\textwidth]{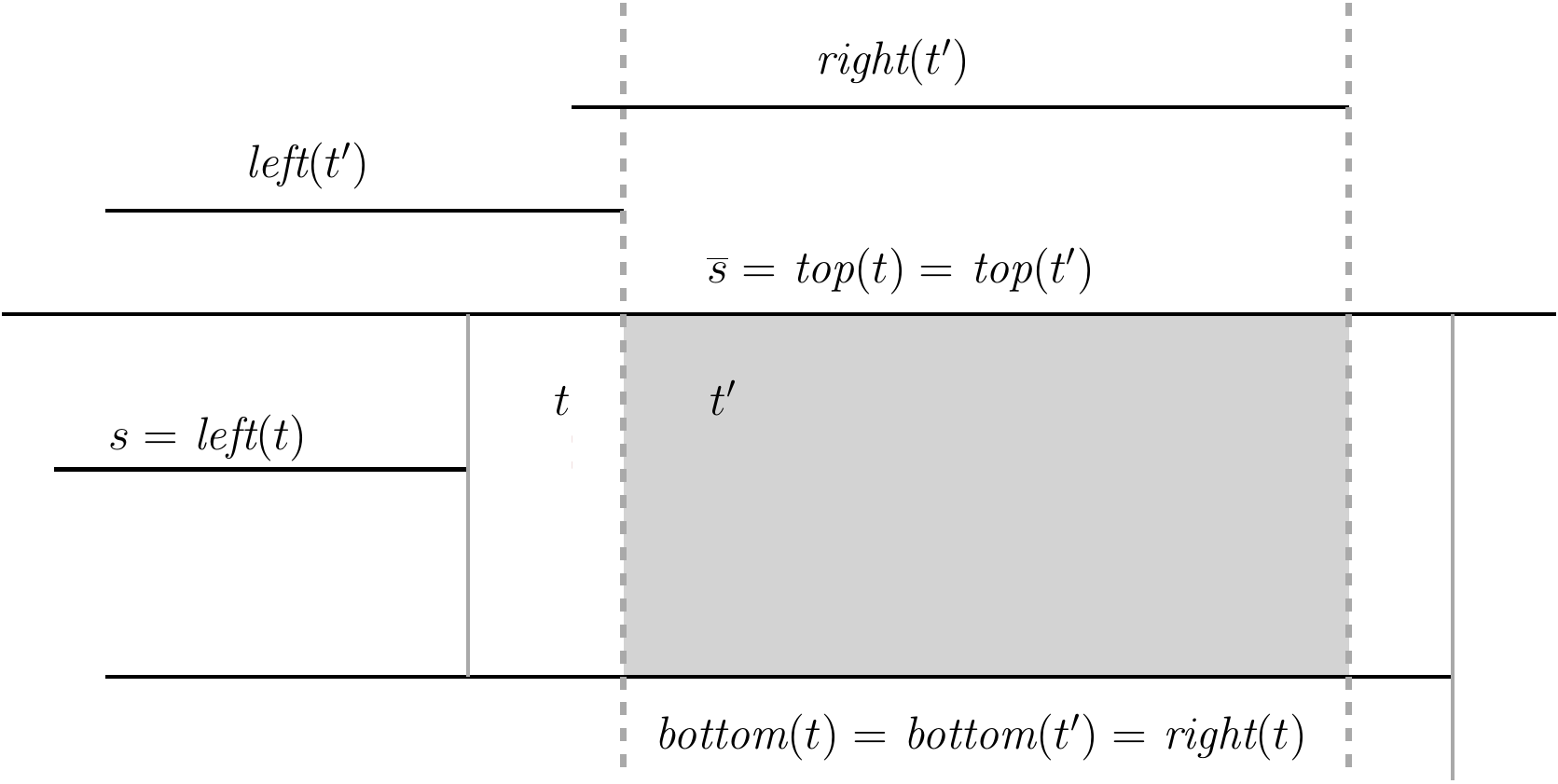} \\
 \caption[Definition of $\phi$]{
   \label{fig:phi}
   \capStyle
   Example configuration of $t'$ covered by $t$.
   In this case $\Right(t)$ is also $\Bottom(t) = \Bottom(t')$ as it
   extends to the left.
   Hence, $s$ must be $\Left(t)$.
   One possible insertion order $\pi$ of the segments causing this configuration is:
   $\Right(t')$, $\Left(t')$, $\Top(t)$, $\Bottom(t)$, $s = \Left(t)$.
   Note that $t'$ is created with the insertion of $\Bottom(t)$,
   whereas $t$ is created afterwards, \ie with the insertion of $s = \Left(t)$.
   Now, $\phi_{t'}^i$ swaps~$s$ with $\overline{s} = \Top(t) = \Top(t')$.
   At its new insertion position in $\phi_{t'}^i(\pi)$ the segment~$\overline{s}$ cannot
   block the vertical walls (dashed) induced by $\Left(t')$ and $\Right(t')$ as
   it did not do so at its earlier position in $\pi$.
   On the other hand, the segment $s$, which is now inserted earlier,
   extends to the left and cannot block these walls at all.
   Hence, according to $\phi_{t'}^i(\pi)$ the trapezoids $t$ and $t'$ are created simultaneously, 
   namely with the insertion of ~$\overline{s}$. 
 } 
\end{figure}

\begin{corollary}
For a random element $\pi$ of $\Pi^{n_{t'} \geq 1}$ the probability that
$n_{t'}(\pi) =i$ for $i>0$ is less than or equal to $1/2^{i-1}$.
\end{corollary}
\begin{proof}
By  Lemma~\ref{lem:injection_lemma} we know that $|\Pi^{n_{t'} > i}| \leq |\Pi^{n_{t'} = i}|$, adding 
$|\Pi^{n_{t'} > i}|$ to both sides we obtain
$$2 |\Pi^{n_{t'} > i}| \leq |\Pi^{n_{t'} > i-1}|,$$
which implies
\begin{eqnarray*}
2^{i-1} |\Pi^{n_{t'} > i-1}| &\leq& |\Pi^{n_{t'} > 0}|,\\
2^{i-1}|\Pi^{n_{t'} \geq i}| &\leq& |\Pi^{n_{t'} \geq 1}|.
\end{eqnarray*}


\noindent
And with $|\Pi^{n_{t'} = i}| \leq |\Pi^{n_{t'} \geq i}|$ we obtain
$$ 2^{i-1}|\Pi^{n_{t'} = i}|  \leq  2^{i-1}|\Pi^{n_{t'} \geq i}| \leq  |\Pi^{n_{t'} \geq 1}|.$$
Thus,
$$\Pr[n_{t'}(\pi) = i | \pi \in \Pi^{n_{t'} \geq 1}] = |\Pi^{n_{t'} = i}| / |\Pi^{n_{t'} \geq 1}| \leq 1/2^{i-1}$$
\myqed
\end{proof}

\begin{corollary}
For a random element $\pi$ of $\Pi^{n_{t'} \geq i}$,
the expected value for $n_{t'}(\pi)$ is constant. \label{cor:expected value}
\end{corollary}
\begin{proof}

\begin{eqnarray*}
\E[n_{t'}(\pi)| \pi \in \Pi^{n_{t'} \geq 1}]
&=   &\sum_{0<i\leq n}i \cdot \Pr[n_{t'}(\pi) = i | \pi \in \Pi^{n_{t'} \geq 1}]\\
&\leq&\sum_{0<i\leq n} i \cdot \frac{1}{2^{i-1}}\\
&=   &2 \cdot\sum_{0<i\leq n} \frac{i}{2^i}\\
&\leq &4
\end{eqnarray*}
\myqed

\end{proof}

According to Corollary~\ref{cor:expected value} the expected number
of different DAG trapezoids that cover $t'$ is not more than $4$.
For every such $t$ that contains $t'$ during the construction
we may only get up to two bouncing nodes.
By also taking into account two additional bouncing nodes that may
occur at the destruction of $t'$,
we can bound $\E[w_{t'}|\delta_{t'}=1]$ by $2\cdot4+2=8$.
Applying this to $\E[w_{t'}|\delta_{t'}=1]$ in Equation~\ref{eq:recursive_4b} yields
\begin{equation}
\label{eq:recursive_5}
O(n\log n) + 8 \sum_{t'\in\Delta}\E[\delta_{t'}],
\end{equation}
which equals:
\begin{equation}
\label{eq:recursive_6}
O(n\log n) + 8\cdot \E[\sum_{t'\in\Delta}\delta_{t'}].
\end{equation}
Clearly, $\E[\sum_{t'\in\Delta}\delta_{t'}] = O(n\log{n})$, since $\E[\sum_{t'\in\Delta}\delta_{t'}]$
is the expected number of trapezoids in $\TrpSrchTree$, proving the following proposition:

\begin{proposition}
\label{pro:recursive_preproc_time}
Let $S$ be a set of $n$ pairwise interior disjoint x-monotone curves
inducing a planar subdivision. Let \G be a DAG of linear size that was constructed by a
randomized incremental insertion.
The expected running time $f(n)$ of the recursive
algorithm executed on \G is $O(n\log n)$.
\end{proposition}

Our main theorem for this section follows immediately by plugging the value
of $f(n)$ obtained in Proposition~\ref{pro:recursive_preproc_time} into
Theorem~\ref{thrm:gnrl_preproc_time}.

\begin{theorem}
Let $S$ be a set of $n$ pairwise interior disjoint x-monotone curves
inducing a planar subdivision. A point location data structure for S, which has
$O(n)$ size and $O(\log n)$ query time in the worst case, can be built in 
expected $O(n \log n)$ time.
\end{theorem}

\subsection{An $O(n\log n)$ Verification Algorithm}
\label{subsec:eff_ver_alg:det}

Let \calTSi denote the trapezoidal map obtained after inserting the
first $i$ curves.
We also use this notation in order to identify
the set of trapezoids of this map.
We denote by \Tall the collection of {\it all} trapezoids created during the
construction of the DAG, including intermediate trapezoids that are
killed by the insertion of later segments.
More formally:
$$\Tall = \bigcup\limits_{i=1}^n \calTSi .$$
Let \Arr{\Tall}
denote the arrangement of all trapezoids in \Tall .
Notice that a face of the arrangement
may be covered by overlapping trapezoids.
The {\it ply} of a point
$p$ in \Arr{\Tall} is defined as the number of
trapezoids in \Tall that cover $p$.
The key to the improved algorithm is the following observation by
Har-Peled~\cite{Har-personal-2012}.
\\

\noindent
\begin{observation} \label{obs:path_len_and_ply_relation}
The length of a path in the DAG for a query point
$q$ is at most three times the ply of $q$ in \Arr{\Tall}.
\end{observation}


It follows that we need to verify that the maximum ply of a point
in \Arr{\Tall} is $c_1 \log n$ for
some constant $c_1 > 0$.
We remark that this ply is established in an
interior of a face of \Arr{\Tall},
since the longest path will always end in a leaf of the DAG,
which, under the general position assumption, represents a trapezoid.
Moreover, for any query point that falls on either a curve or an endpoint
of the initial subdivision
the search path will end in an internal node of the DAG.
If, on the other hand, the query point $q$
falls on a vertical edge of a trapezoid,
the search path for $q$ will be identical to a path
for a query point in a neighboring trapezoid.
Therefore, we consider the boundaries of the trapezoids as open.

Since the input curves are interior
pairwise disjoint,
according to the separation property deduced from \cite{GY-TSR-80},
one can define a total order on the curves;
see more details in Subsection~\ref{subsubsec:eff_ver_alg:reduction} below.
This order allows us to apply a modified version
of an algorithm by Alt and Scharf~\cite{as-cdaaa-13},
which originally detects the
maximum ply in an arrangement of $n$
axis-parallel rectangles in $O(n \log n)$ time.
Recall that we only apply this verification algorithm on DAGs of
linear size.

We would like to describe a linear space
algorithm with $O(n\log n)$ running time
for computing the ply of an arrangement of
open trapezoids with the following properties:
their bases are $y$-axis parallel (vertical walls) and if
the top or bottom curves of two different trapezoids intersect
not only in a joint endpoint
then the two curves overlap completely in their joint $x$-range.
The ply of such an arrangement is the maximum
number of trapezoids containing a common point,
that is, we are only interested in points located
in faces of this arrangement.
In Appendix~\ref{appendix_alt_scharf} we restate the algorithm by Alt \& Scharf~\cite{as-cdaaa-13}
such that the general position assumption can be dropped.
\arxivdcg{The restated algorithm
is more general than what we essentially need
as it can handle rectangles
with independently open or closed boundaries,
while we are only interested in open rectangles.}{}
The algorithm constructs a balanced binary tree representing the possible $x$-intervals.
It then performs a vertical sweep, recording the events of creation and destruction of a rectangle.
The data is kept in the tree nodes, and a final traversal pushes the collected information to the leaves.
The maximal ply will appear in one of the leaves.

Next, we define a reduction from
the collection of open trapezoids \Tall to a collection \Rall of
open axis-parallel rectangles such that the maximum ply in \Arr{\Rall}
is the same as the maximum ply in \Arr{\Tall}.
Using this reduction we can finally
%
describe a modification for the
restated algorithm such that it can compute the ply of
the arrangement of all trapezoids
created during the
construction of the DAG.

\subsubsection {A Ply Preserving Reduction}
\label{subsubsec:eff_ver_alg:reduction}

Let \TcollReduc be a collection of open trapezoids with $y$-axis parallel bases
with the following property:
if the top or bottom curves of two different trapezoids intersect
not only in joint endpoints
then the two curves overlap completely in their joint $x$-range.
Let \Arr{\TcollReduc} denote the arrangement
of the trapezoids in \TcollReduc . Notice that 
each arrangement face can be covered by overlapping trapezoids.
We describe a reduction from \TcollReduc to
\RcollReduc, where \RcollReduc is a collection of
axis-parallel rectangles, such that the maximum ply
in \Arr{\RcollReduc} equals to the
maximum ply in \Arr{\TcollReduc}.

In order to define the reduction
we need to have a total order~$<$
on the non-vertical curves of the
trapezoids in~\TcollReduc,
such that one can translate the curves
one by one
according to this order
to $y=-\infty$
without hitting other curves that
have not been moved yet.
Guibas~\&~Yao~\cite{GY-TSR-80} defined
an acyclic relation $\prec$ on a set $C$ of $n$ interior disjoint
$x$-monotone curves as follows:

\begin{definition}
For two such curves $cv_i, cv_j \in C$,
let the open interval $(a,b)$ be the $x$-range of $cv_i$
and the open interval $(c,d)$ be the $x$-range of $cv_j$.\\
If $x\text{-range}(cv_i) \bigcap x\text{-range}(cv_j) \neq\emptyset$ then:\\
\mbox{\ \ } $cv_i \prec cv_j \Leftrightarrow
 cv_i(x) < cv_j(x) \text{ for some } x \in x\text{-range}(cv_i) \bigcap x\text{-range}(cv_j)$.
\end{definition}

As a matter of fact, their definition is more specific,
in a way that the relation $cv_i \prec cv_j$ exists only
if $cv_i$ is the first curve encountered by $cv_j$
in their joint $x$-range
while translating $cv_j$ to $y=-\infty$.
In~\cite{GY-TSR-80} it is also mentioned that $\prec^+$, which is the transitive closure
of $\prec$, is a partial order (as it allows transitivity).
This partial order $\prec^+$ can be extended to a total order $<$ in many ways.
One possible extension is defined as follows:

\begin{definition}
\label{def:ord}
Let $C$ be a set of interior disjoint $x$-monotone curves.
For two curves $cv_i, cv_j \in C$,
let the open interval $(a,b)$ be the $x$-range of $cv_i$
and the open interval $(c,d)$ be the $x$-range of $cv_j$.\\
The total order $<$ on $C$ is defined as follows:\\
\mbox{\ \ } $cv_i < cv_j \Leftrightarrow (cv_i \prec^+ cv_j)$ or $(\neg(cv_j \prec^+ cv_i)$ and $(cv_i$ \emph{left} $cv_j))$ \\
where
$(cv_i$ \emph{left} $cv_j)$ is true if the $x$-value of the left endpoint of $cv_i$ is
less than the $x$-value of the left endpoint of $cv_j$.
\end{definition}

Clearly, if $cv_i \prec^+ cv_j$ is true then $cv_i < cv_j$ is true as well.
If for two different curves $cv_i, cv_j$ the expression $cv_j \prec^+ cv_i$
is true then obviously $cv_i \prec^+ cv_j$ is false
and also the right-hand side expression in the ``or" phrase is false, since~%
$\neg(cv_j \prec^+ cv_i)$ is false. Therefore, $cv_i < cv_j$ is also false.
If the partial order $\prec^+$ does not say anything about $cv_i$ and $cv_j$
then both $(cv_i \prec^+ cv_j)$ and~$(cv_j \prec^+ cv_i)$ are false.
Thus, $cv_i < cv_j$ will be true only if $(cv_i$ \emph{left} $cv_j)$ is true.

Ottmann~\&~Widmayer~\cite{OW-TSLS-83} presented a one-pass
$O(n\log n)$ time algorithm for computing
$<$, as in Definition~\ref{def:ord}, using linear space.
Their algorithm performs a sweep using a horizontal line
from bottom to top which stops at each endpoint of a curve.
The data structure maintained by the algorithm represents
the curves encountered so far in reverse order.
When a bottom endpoint of a curve is met the curve is inserted into an auxiliary structure
holding the active curves only. A curve is removed from the auxiliary structure when its
top endpoint is met by the sweep line.
Since we would like to translate the curves to $y=-\infty$,
then we should only require the curves to be $x$-monotone.
In addition, we can require the curves to be interior disjoint,
rather than completely disjoint.

\begin{definition}
Let Rank$: C\rightarrow\{1,...,n\}$ denote a function
returning the rank of a given $x$-monotone curve $cv \in C$
when sorting $C$ according to the total order~$<$ as defined above. 
\end{definition}

\begin{definition}
\label{def:reduction}
We define a reduction from \TcollReduc to
\RcollReduc as follows;
Every trapezoid~$\tred~\in~\TcollReduc$ is
reduced to a rectangle $\rred \in \RcollReduc$,
such that:
\begin{itemize}
  \item \tred and \rred have the same $x$-range,\\
  i.e., $($left$(\tred) =$ left$(\rred))$ and
  $($right$(\tred) =$ right$(\rred))$,
  where left and right denote the left $x$-value
  and the right $x$-value of \tred (or \rred), respectively.
  \item top$(\rred)$ and bottom$(\rred)$ 
  lie on $y=$Rank$($top$(\tred))$ and
  $y=$Rank$($bottom$(\tred))$, respectively.
\end{itemize}
\end{definition}

Definition~\ref{def:reduction} provides a mapping from \TcollReduc to \RcollReduc,
such that $r$ is the rectangular region corresponding to $t$.
\arxivdcg{
We will now show that this mapping is bijective.
One can partition the plane into
vertical slabs by passing a vertical line through every endpoint
of the subdivision, and then partition each slab into regions by
intersecting it with all possible curves in the subdivision.
This defines a decomposition of the plane into at most~$2(n+1)^2$
regions (see~\cite{CG-alg-app}, for example).

\begin{lemma}
\label{lemma:regions_bijection}
Let Regions(arr) denote the collection of regions
of an arrangement arr, as defined above.
For any region $a_t \in$ Regions(\Arr{\TcollReduc})
let $a_r \in$ Regions(\Arr{\RcollReduc}) be the corresponding
rectangular region to $a_t$.
The collection Regions(\Arr{\RcollReduc})
of all such rectangular regions spans the plane.
\end{lemma}

\begin{proof}
Trivial. The slabs remain the same and
within each slab the rectangular regions remain adjacent.
\myqed
\end{proof}

\begin{lemma}
\label{lemma:reduction_from_trpz_to_rect}
Let $a_t \in$ Regions(\Arr{\TcollReduc}) be a region
and let $a_r \in$ Regions(\Arr{\RcollReduc}) be the
corresponding rectangular region to $a_t$.
The number of rectangles in \RcollReduc that cover $a_r$
is at least the number of trapezoids in \TcollReduc that cover $a_t$.
In other words,
for every $\tred \in \TcollReduc$ that covers $a_t$
its corresponding rectangle $\rred \in \RcollReduc$ covers $a_r$.
\end{lemma}

\begin{proof}
Let $\{\tred_1,\tred_2,...,\tred_m\}\subseteq \TcollReduc$
be the set of trapezoids, ordered by creation time,
such that for every $i\in\{1,...,m\}$, $\tred_i$ covers $a_t$.
Let $\{\rred_1,\rred_2,...,\rred_m\}\subseteq \RcollReduc$
be the set of corresponding
rectangles, such that $\rred_i$ corresponds to $\tred_i$ for
$i \in \{1,...,m\}$.
For any $\tred_i$, since $\tred_i$ covers $a_t$
we get that $x\text{-range}(a_t)\subseteq x\text{-range}(\tred_i)$.
By Definition~\ref{def:reduction} the
$x$-ranges remain the same after the reduction,
and therefore $x\text{-range}(a_r) \subseteq x\text{-range}(\rred_i)$.
Since $\tred_i$ covers $a_t$ then we also get that
in the shared $x$-range top$(\tred_i)$ is above or on top$(a_t)$
and bottom$(\tred_i)$ is below or on bottom$(a_t)$.
According to Definition~\ref{def:reduction},
it immediately follows that
Rank$($top$(\tred_i)) \geq$ Rank$($top$(a_t))$.
In other words, top$(\rred_i)$ is above or on top$(a_r)$.
Similarly, bottom$(\rred_i)$ is below or on bottom$(a_r)$.
We conclude that $\rred_i$ covers $a_r$.
\myqed
\end{proof}

\begin{lemma}
\label{lemma:reduction_from_rect_to_trpz}
Let $a_r\in$ Regions$(\Arr{\RcollReduc})$ be a rectangular region,
whose corresponding region is $a_t \in$ Regions$(\Arr{\TcollReduc})$.
The number of trapezoids in \TcollReduc that cover $a_t$
is at least the number of rectangles in \RcollReduc that cover $a_r$.
In other words,
for every $\rred \in \RcollReduc$ that covers $a_r$
its corresponding trapezoid $\tred \in \TcollReduc$ covers $a_t$.
\end{lemma}

\begin{proof}
Let $\{\rred_1,\rred_2,...,\rred_m\}\subseteq \RcollReduc$
be the set of rectangles,
such that for every $i\in\{1,...,m\}$, $\rred_i$ covers $a_r$.
Let $\{\tred_1,\tred_2,...,\tred_m\}\subseteq \TcollReduc$
be the set of corresponding
trapezoids, such that $\tred_i$ corresponds to $\rred_i$ for
$i \in \{1,...,m\}$.
Proving that for any $i \in \{1,...,m\}$, $\tred_i$ covers $a_t$,
is done symmetrically to the proof of Lemma~\ref{lemma:reduction_from_trpz_to_rect}.
\myqed
\end{proof}

Combining Lemma~\ref{lemma:reduction_from_trpz_to_rect} and
Lemma~\ref{lemma:reduction_from_rect_to_trpz}
we conclude that the number of trapezoids in \TcollReduc
that cover a region $a_t$ equals to the number
of rectangles in \RcollReduc that cover $a_r$, which is
the corresponding region to $a_t$.
The covering rectangles are the reduced trapezoids
in the set of trapezoids covering $a_t$.
Since both Regions$(\Arr{\TcollReduc})$
and Regions$(\Arr{\RcollReduc})$ span the plane
(Lemma~\ref{lemma:regions_bijection}),
In summary, we obtain the following theorem.
}
{
In Appendix~\ref{appendix_reduction_bijection} we show that this mapping is bijective.
We show there that the number of trapezoids in \TcollReduc
that cover a region $a_t$ equals to the number
of rectangles in \RcollReduc that cover $a_r$, which is
the region corresponding to $a_t$.
In summary, we obtain the following theorem.
}

\begin{theorem}
\label{theorem:reduction_correctness}
Let \TcollReduc be a collection of
open trapezoids with the following properties:
their bases are $y$-axis parallel (vertical walls) and if
the top or bottom curves of two different trapezoids intersect
not only in joint endpoints
then the two curves overlap completely in their joint $x$-range.
Let \Arr{\TcollReduc} denote the arrangement
of the trapezoids in~\TcollReduc. Notice
that each arrangement face can be covered by
overlapping trapezoids.
\TcollReduc can be reduced to a collection of open
axis-parallel rectangles~\RcollReduc,
such that the maximum ply in \Arr{\RcollReduc} equals
to the maximum ply in \Arr{\TcollReduc}.
\end{theorem}

\subsubsection{ Modification of Alt \& Scharf}
\label{subsubsubsec:eff_ver_alg:depth_general}

Based on the correctness of the reduction
described above 
we can extend the basic algorithm by Alt \& Scharf~\cite{as-cdaaa-13}
to support not only collections of axis-aligned rectangles but
also collections of open trapezoids with $y$-axis parallel bases and
non-intersecting top and bottom boundaries, if they intersect
not only in joint endpoints
then they overlap completely in their joint $x$-range.
The only part of the basic algorithm that should change
is the top-to-bottom sweep.
More precisely, the simple predicate that is used for sorting the $y$-events
should be replaced with a new predicate that
compares according to the reverse order of~$<$, 
as given in Definition~\ref{def:ord}.
The total order~$<$ can be computed in a preprocessing phase
using the algorithm in~\cite{OW-TSLS-83}.

Notice that for simplicity
we assumed that no two distinct endpoints
in the original subdivision have
the same $x$-value.
However, if this is not the case,
lexicographical comparison can be used on the endpoints
of the curves in order to define the order
of the induced vertical walls.

\subsection{Summary}
\label{subsec:eff_ver_alg:static_bounds}

The two algorithms described in Subsection~\ref{subsec:eff_ver_alg:rec}
and Subsection~\ref{subsec:eff_ver_alg:det} can be used
for defining
efficient construction algorithms for static settings,
according to the scheme presented in Theorem~\ref{thrm:gnrl_preproc_time}.

Using either verification algorithm, 
a construction algorithm with expected $O(n\log {n})$ running time is obtained,
implying the following main contribution of the paper:

\begin{theorem}
\label{thrm:new}
Let $S$ be a set of $n$
pairwise interior disjoint $x$-monotone curves
inducing a planar subdivision.
A \PL data
structure for $S$, which has~$O(n)$ size and $O(\log{n})$ query time
in the worst case, can be built in expected~$O(n\log{n})$ time.
\end{theorem}


 \section{Conclusions and Open Problems}
\label{sec:conclusion}

In this work we have described an optimal variant of a known algorithm for point location:
the randomized incremental construction of the trapezoidal map. 
This fundamental 
point location algorithm supports general
$x$-monotone curves and guarantees logarithmic query time and linear space.
Previously with such guarantees, the expected construction time of the randomized search structure
was $O(n\log^2{n})$, as was mentioned in~\cite{CG-alg-app}.
Their construction algorithm uses an auxiliary algorithm
for verifying that the maximal query path length is logarithmic,
whose expected time complexity is $O(n\log^2{n})$.
The latter
dominates the overall construction complexity.
However, we have proposed two novel verification algorithms%
---either of which could
be used instead when constructing the search structure.
These two efficient verification algorithms allow an expected $O(n\log{n})$ construction time while having the same guarantees.

The two possible verification algorithms
we have described can both be plugged into the suggested construction scheme.
The first algorithm operates directly on the DAG~\G and has a recursive nature.
Its analysis relies on a bijection
between all search paths in~\G
and those of the trapezoidal search tree~\TrpSrchTree.
The second algorithm has a deterministic $O(n \log {n})$ time complexity,
and it is based on the computation of the maximal ply of all
trapezoids that existed during the construction of the DAG~\G.
While the latter is deterministic,
the former
does not require the construction of any other auxiliary structures
and only uses the already constructed DAG.

Another contribution of this work, which in fact triggered the
entire project,
is the study of the fundamental difference
between the length~\lqpl of the longest search path
and the DAG depth~\depth, which is the length of the longest path in the constructed DAG.
Clearly, efficiently computing the value of~\lqpl is not trivial,
whereas~\depth can be easily accessed.
We have clarified why the two entities
are not trivially interchangeable and proved that
the worst-case ratio of $\depth/\lqpl$
is in~$\Theta(n/\log n)$.

One major open problem, which is of theoretic interest, is whether it is
sufficient to simply check the~\depth during the construction,
as it is in fact done in the current \cgal
implementation~\cite{HemmerKH12}, and still expect a constant number of rebuilds.
In other words, can we still expect a constant number
of rebuilds if we just rely on~\depth, which is only an upper bound of~\lqpl?



\section{Acknowledgement}
\label{sec:ack}
We thank Haim Kaplan for many helpful discussions. 


\bibliography{how_to_cite_cgal,bibliography,links}



\bibliographystyle{plain}


\newpage
\appendix
\section{The Trapezoidal Search Tree~\TrpSrchTree}
\label{appendix_trapezoidal_search_tree}

We have shown a bijection
between all search paths in the DAG~\G
and those of the trapezoidal search tree~\TrpSrchTree (see Section~\ref{sec:relation_G_and_T}).
Using this bijection we were able to devise an efficient
recursive verification algorithm that is described in Subsection~\ref{subsec:eff_ver_alg:rec}.
The analysis of the algorithm relies on
the expected $O(n\log{n})$ size of~\TrpSrchTree.
Even though this bound seems to be folklore,
we have not found the source in which it is actually proven.
Therefore, we provide here a proof.

\begin{definition}
Let $S$ be a set of $n$
pairwise interior disjoint $x$-monotone curves
inducing a planar subdivision.
The trapezoidal search tree~\TrpSrchTree
for $S$ is a full binary tree constructed as the DAG~\G using the same insertion order
while skipping the merge step.
\end{definition}

The following lemma bounds the expected size of \TrpSrchTree to be $O(n \log n)$.
 \\
 \\
\noindent\textbf{Lemma 2} \emph{Let $S$ be a set of $n$
pairwise interior disjoint $x$-monotone curves
inducing a planar subdivision.
The expected number of leaves
in the trapezoidal search tree \TrpSrchTree,
which is constructed as the DAG but without merges,
is~$O(n\log{n})$.}

\begin{proof}
We would like to bound the expected number
of leaves in~\TrpSrchTree, namely,
the expected number
of trapezoids in the decomposition without merges.
To ease the argument, we can symbolically shorten every curve at its two endpoints
by an arbitrarily small value $\varepsilon>0$.
In other words, if a curve $cv$ 
has an $x$-range $(a,b)$,
then the shortened curve will have an $x$-range $(a+\varepsilon, b-\varepsilon)$.
The curves of the updated subdivision are now completely disjoint.
This operation only gives rise to new artificial trapezoids.
Hence, it is sufficient to bound the expected number of trapezoids in this subdivision
of shortened curves.

Now we would like to bound the number of trapezoids in the set
of shortened curves.
It is clearly bounded by the number of vertical edges plus 1.
First, consider the vertical line~$W$ through one endpoint of a curve $cv$.
$W$ is intersected by $m$ curves.
Suppose that $cv$ is the $i$th inserted curve among these $m$ curves.
The $i-1$ already inserted curves partition $W$ into $i$ intervals.
However, we are only interested in the interval $I$ containing the
endpoint of~$cv$, as it will
appear in the final structure.
Curves inserted after~$cv$ may split~$I$.
The expected number of intersections in $I$
(including the endpoint of~$cv$) is~$O((m-i)/i)$.
The probability that $cv$ will be inserted $i$th is $\frac{1}{m}$.
Summing up over all possible insertion orders we get that the expected number
of intersections in $I$ is $\sum\limits_{i=1}^{m} \frac{1}{m}\cdot\frac{(m-i)}{i} = O(\log{m})$.
However, since $m \leq n$,
this number can be bounded $O(\log{n})$.
There are $O(n)$ vertical walls, giving a total of expected $O(n\log{n})$ intersections.
Thus, the expected number of vertical edges is $O(n\log{n})$ as well,
and, clearly, this is also the expected number of leaves in the tree.

\ignore{
First, consider the vertical line~$W$ through one endpoint of the 
$i$th inserted curve.
$W$ is intersected by $n$ curves, in the worst case.
The $i-1$ already inserted curves partition $W$ into $i$ intervals.
However, we are only interested in the interval $I$ containing the
endpoint of the $i$th curve, as it will
appear in the final structure.
Curves inserted after the $i$th curve may split $I$.
The expected number of intersections in $I$
(including the endpoint of the $i$th curve) is $O((n-i)/i)$.
Summing up over all vertical walls gives a total of expected $O(n\log{n})$ intersections.
Thus, the expected number of vertical edges is $O(n\log{n})$ as well,
and, clearly, this is also the expected number of leaves in the tree.}
\myqed
\end{proof}

\clearpage


\section{Bijection between Search Paths in~\G and~\TrpSrchTree}
\label{appendix_bijection}
We provide here the full case analysis
that is needed for proving Proposition~\ref{prop:identical_search_paths_per_query}.
\\


\noindent
{\bf Proposition 1}
{\em Let $S$ be a set of $n$
pairwise interior disjoint $x$-monotone curves
inducing a planar subdivision.
Let \G and~\TrpSrchTree be the DAG and the trapezoidal search tree
created using the same permutation of the curves in $S$, respectively.
There exists a canonical bijection among all search paths in~\G and those of~$\TrpSrchTree$,
that is, for any query point $q$, the corresponding search paths for $q$
in~\G and~\TrpSrchTree are identical up to bouncing nodes.}
%
\begin{proof}
Let~$q$ be a query point and~$t$ and~$t'$ be the leaf trapezoids containing~$q$
in~$\G$ and~\TrpSrchTree, respectively. Obviously, the top and bottom curves
of~$t$ and~$t'$ are identical and~$t$ covers~$t'$ since merges are only allowed
in~\G.
Suppose that while searching for~$q$ we maintain the history-interval of possible~$x$-values\footnote{Since we do not require the
points to be in general position we consider the lexicographic order of the points.
Therefore, this interval which is referred to as $x$-interval is essentially defined by the $x$,$y$ coordinates of two points}.
At the end of the search in~\TrpSrchTree this interval is obviously identical
to the~$x$-range of~$t'$.
We show by induction the bijection between the two search paths for~$q$ and in
fact we also show that the history intervals maintained while searching in~\G
and~\TrpSrchTree are identical, i.e., the history interval that is eventually
obtained by the search in~\G is identical to the~$x$-interval of~$t'$.

Let~$\G_{i}$, $\TrpSrchTree_{i}$ denote the DAG and the trapezoidal search tree
after the first~$i$ curves were inserted, respectively.
We denote by~$t_{i}$ and~$t_{i}'$ the trapezoids containing the query point~$q$
in~$\G_{i}$ and~$\TrpSrchTree_{i}$, respectively.
Let~$(a_{i}, b_{i})$ and~$(a'_{i}, b'_{i})$ denote the~$x$-intervals
of~$t_i$ in~$\G_i$ and~$t'_i$ in~$\TrpSrchTree_i$, respectively.
The base case is trivial since~$\G_1 = \TrpSrchTree_1$.
Now suppose that the statement holds for~$i-1$.
We show that it holds for~$i$ as well.
The~$i$th curve~$cv_i(p_i,q_i)$ is now inserted into both~$\G_{i-1}$ and~$\TrpSrchTree_{i-1}$.
The basic argument is as follows. For both endpoints of~$cv_i$ there are essentially
three cases:
(i) the point is outside $t_{i-1}$ and $t'_{i-1}$: the point has no effect on both paths.
(ii) the point is inside $t_{i-1}$ and $t'_{i-1}$: the point shows up as a normal node in both paths
and the history intervals are updated accordingly.
(iii) the point is inside $t_{i-1}$ but not in $t'_{i-1}$: the point has no effect on the search path
in $T_i$ while it may show up on the search path in~$\G_i$, but only as a bouncing node, i.e., the
history interval remains unchanged.
Figure~\ref{fig:induction_all_full} shows the 15 possible positions to insert $cv_i$ with respect to
$t_{i-1}$ and $t'_{i-1}$.
We denote the five different vertical slabs (regions) depicted in the figure by $A,B,C,D$, and $E$.
Note that regions $B$ and $D$ may have zero width.

\begin{figure} [h] 
    \centering
    \vspace{-10pt}
    \includegraphics[width=0.45\textwidth]{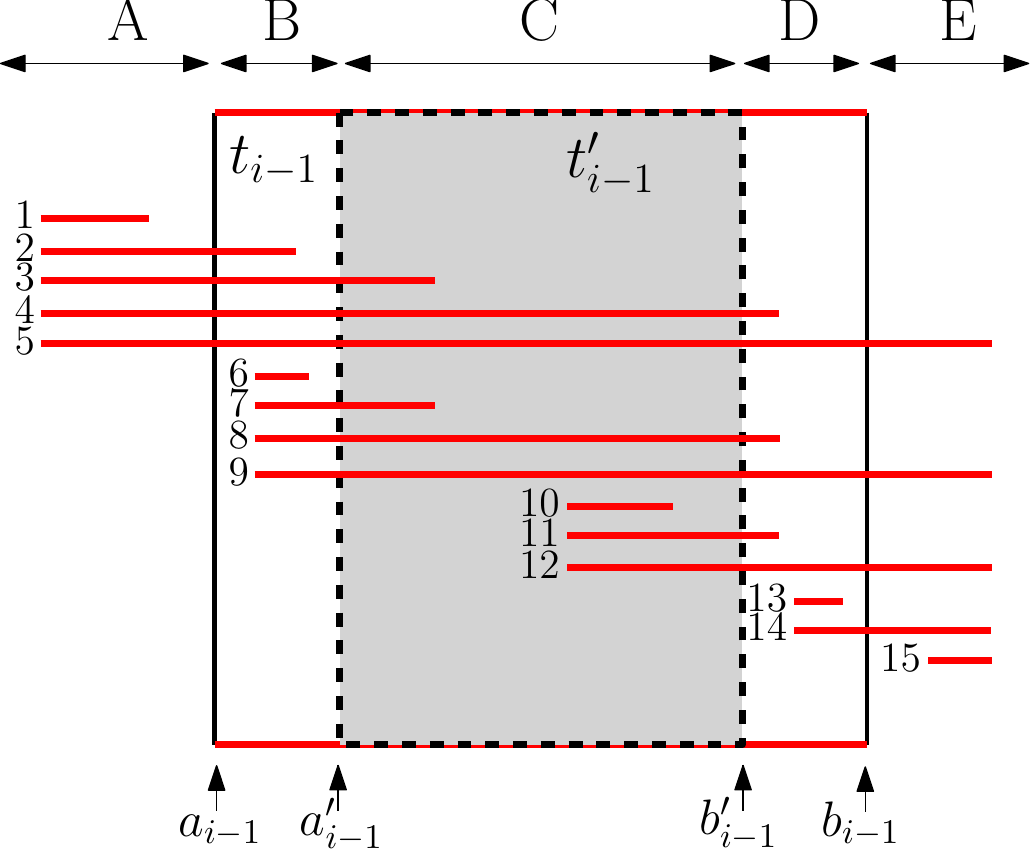}
   \caption[Optional Next-Step in the Induction]{
   \capStyle{Possible positions for $cv_i$ in relation to trapezoid $t_{i-1}$ in $\G_{i-1}$, which covers trapezoid $t'_{i-1}$ in $\TrpSrchTree_{i-1}$. }}
        \label{fig:induction_all_full}
\end{figure}
\noindent

For ease of reading we group the optional positions according to
the region containing $p_i$ as follows:
\begin{itemize}
    \item $p_i$ is located at region~$A$ (positions 1-5 in Figure~\ref{fig:induction_all_full}).
         For these positions $p_i$ will not be added to the path to $q$ in either $\G_i$ or $\TrpSrchTree_i$ (case (i)).
         We now distinguish the different cases depending on the position of $q_i$.
         \begin{itemize}
            \item Position~1: $q_i$ lies in region~$A$ as well and, therefore, 
                will not affect both paths (case (i)).
                Clearly,
                $t_i = t_{i-1}$ and $t_i' = t_{i-1}'$.
            \item Position~2: $q_i$ lies in region~$B$ (case (iii)).
                $q_i$ will not be added to the path to $q$ in~$\TrpSrchTree_i$.
                However, it will be added to the path in $\G_i$ but only as a bouncing node for this path.
                The $x$-interval maintained during the search in $\G_i$ will not be affected.
            \item Position~3: $q_i$ lies in region~$C$ (case (ii)).
                The search paths for $q$ in both $\G_i$ and $\TrpSrchTree_i$
                will include $q_i$.
                If $q$ is in the $x$-range of $cv_i$ then an additional
                internal node representing $cv_i$
                will appear in the path for $q$ in both structures.
                The interval $(a_{i}', b_{i}')$ in such a case would be
                $(a_{i-1}',q_i)$.
                If, on the other hand, $q$ is to the left of $q_i$ then the
                new interval would be $(q_i,b_{i-1}')$.
            \item Position~4: $q_i$ lies in region~$D$ (case (iii)).
                Similar to the case where $q_i$ lies in region~$B$.
                In addition, since $cv_i$ intersects~$t_{i-1}'$
                completely, an internal node $cv_i$ will be added to both structures.
            \item Position~5: $q_i$ lies in region~$E$ and, therefore,
                will not affect both paths (case (i)).
                Since $cv_i$ intersects~$t_{i-1}'$
                completely, an internal node $cv_i$ will be added to both structures.
         \end{itemize}
    \item $p_i$ is located in region~$B$ (positions 6-9 in Figure~\ref{fig:induction_all_full}). 
         For these positions~$p_i$ will not be added to the path to $q$ in~$\TrpSrchTree_i$.
         However, it will be added to the path in~$\G_i$ but only as a bouncing node for this path, since it is not contained in $(a_{i-1}',b_{i-1}')$ (case (iii)).
         We now distinguish the several cases depending on the position of $q_i$.
         \begin{itemize}
            \item Position~6: $q_i$ lies in region~$B$ (case (iii)).
                In such a case $q_i$ will be added to the query path
                to $q$ as a bouncing node in $\G_i$,
                but will not affect the interval maintained during the search since $q_i$ is not
                contained in $(a_{i-1}',b_{i-1}')$.
                The query path to $q$ in $\TrpSrchTree_i$ will not change, since $cv_i$ does not intersect~$t_{i-1}'$.
            \item Position~7: $q_i$ lies in region~$C$ (case (ii)).
                The search paths for $q$ in both $\G_i$ and $\TrpSrchTree_i$
                will include $q_i$.
                If $q$ is in the $x$-range of $cv_i$ then an additional
                internal node representing $cv_i$
                will appear in the path for $q$ in both structures.
                The interval $(a_{i}', b_{i}')$ in such a case would be
                $(a_{i-1}',q_i)$.
                If, on the other hand, $q$ is to the left of $q_i$ then the
                new interval would be $(q_i,b_{i-1}')$.
            \item Position~8: $q_i$ lies in region~$D$ (case (iii)).
                Similar to the case where $q_i$ lies in region~$B$.
                In addition, since $cv_i$ intersects~$t_{i-1}'$
                completely, an internal node $cv_i$ will be added to both structures.
            \item Position~9: $q_i$ lies in region~$E$ and, therefore, will not affect both paths (case (i)).
                Since $cv_i$ intersects~$t_{i-1}'$
                completely, an internal node $cv_i$ will be added to both structures.
         \end{itemize}
    \item $p_i$ is located inside in region~$C$ (positions 10-12 in Figure~\ref{fig:induction_all_full}). 
         In these positions $p_i$ will be added to the search path of a query point~$q$
         that lies in $t_{i-1}'$ both in~$\G_i$ and in~$\TrpSrchTree_i$,
         since it is contained in $(a_{i-1}',b_{i-1}')$ (case (ii)).
         We now distinguish the several cases depending on the position of $q_i$.
         \begin{itemize}
            \item Position~10: $q_i$ lies in region~$C$ (case (ii)).
                $cv_i$ is contained completely in region~$C$.
                The same internal nodes, depending on the position of~$q$, will
                be added for both search structures.
            \item Position~11: $q_i$ lies in region~$D$ (case (iii)).
                If the query point~$q$ is located to the left of $p_i$ then
                no new node (other than $p_i$) will be added to the search paths
                of $q$ in both~$\G_i$ and~$\TrpSrchTree_i$.
                If, on the other hand,~$q$ is in the $x$-range of $cv_i$
                then~$q_i$ will be added to the path as a bouncing node in $\G_i$,
                but will not appear in~$\TrpSrchTree_i$.
                In addition the paths in the two structures will be added with a
                node representing~$cv_i$.
            \item Position~12: $q_i$ lies in region~$E$ and, therefore,
                will not affect both paths (case (i)).
                Depending on the location of~$q$, an internal node~$cv_i$ may be added
                to the paths in both structures.
         \end{itemize}
      \item $p_i$ is located in region~$D$ (positions 13-14 in Figure~\ref{fig:induction_all_full}). 
         For these positions~$p_i$ will not be added to path to $q$ in~$\TrpSrchTree_i$.
         However, it will be added to the path in $\G_i$ but only as a bouncing node for this path (case (iii)).
         In both positions $q_i$ is to the right of $p_i$ and is, therefore,
         blocked by $p_i$ for query points that lie in~$t_{i-1}'$ and will not be added
         to the search paths of such points.
       \item $p_i$ is located to the right of $t_{i-1}$, that is, in region~$E$ (case (i)).
         In Figure~\ref{fig:induction_all_full} the relevant position is~15.
         Both $p_i$ and $q_i$, which is located to the right of $p_i$, will not affect the search paths for $q$ in both structures.
\end{itemize}

\noindent In each of these 15 different cases, whenever $p_i$ or $q_i$ are only added to $\G_i$,
the node will appear as a bouncing node for the path to $q$.

\myqed
\end{proof}

\clearpage

\section{An Algorithm for Computing the Ply of an Arrangement of Axis-aligned Rectangles}
\label{appendix_alt_scharf}

The algorithm of Alt \& Scharf~\cite{as-cdaaa-13}
is an $O(n\log n)$ algorithm that computes
the maximum ply of an arrangement of axis-aligned rectangles
in general position, using $O(n)$ space.
{
We present a minor modification,
which does not assume general position,
i.e., rectangles may share boundaries.
Moreover, it can consider each of the four boundaries
of a rectangle as either 
belonging to the rectangle or not; we call these closed or open boundaries, respectively.

Given a set of $n$ axis-aligned rectangles,
let $x_1,x_2,...,x_k$, $k \leq 2n$ be the sorted set of $x$-coordinates
of the vertical sides of the rectangles.
The ordered set of intervals $\cal I$ is defined as follows;
for $i\in{1,2,...,k-1}$,
the~$2(i-1)$th and~$2(i-1)+1$st intervals in the set $\cal I$
are $[x_{i},x_{i}]$ and $(x_{i},x_{i+1})$, respectively.
The last interval is $[x_k,x_k]$.
A balanced binary tree $T$ is then constructed,
holding all intervals in $\cal I$ in its leaves,
according to their order in $\cal I$.
An internal node represents the union of
the intervals of its two children,
which is a contiguous interval.
In addition, each internal node $v$ stores in a variable~$v.x$
the $x$-value of the merge point between the intervals
of its two children. Since we extended the algorithm
to support both open or closed boundaries, internal nodes
also maintain a flag indicating whether the merge point
is to the left or to the right of the $x$-value.

According to the description of the algorithm in~\cite{as-cdaaa-13},
a sweep is performed using a horizontal line from
$y=\infty$ to $y=-\infty$.
The sweep-line events occur when a rectangle starts or ends,
i.e., when top or bottom boundary of a rectangle is reached.
Since the rectangles are not in general position,
several events may share the same $y$-coordinate.
In such a case, the order of event processing in
each $y$-coordinate is as follows:
\begin{enumerate}
  \item Closing rectangle with open bottom boundary events.
  \item Opening rectangle with closed top boundary events.
  \item Closing rectangle with closed bottom boundary events.
  \item Opening rectangle with open top boundary events.
\end{enumerate}
The order of event processing within each of these
four groups in a specific $y$-coordinate is not important.

The basic idea of the algorithm is that
each sweep event updates the 
leaves of the tree $T$
that span the intervals that are covered by the event.
Therefore, each leaf holds a counter $c$ for the
number of covering rectangles in the current position
of the horizontal sweep line.
In addition, each leaf maintains in
a variable $c_m$ the maximal number of
covering rectangles for this leaf seen so far.
Clearly, the maximal coverage of an interval
is the maximal $c_m$ of all leaves.
The problem with this na\"{i}ve approach
is that one such update can already
take $O(n)$ time.
Therefore, the key idea of~\cite{as-cdaaa-13}
is that when updating an event of a rectangle
whose $x$-range is $(a,b)$,
one should follow only two paths;
the path to $a$ and the path to $b$.
The nodes on the path should hold the
information of how to update the
interval spanned by their children.
In the end of the update the union of
intervals spanned by the updated nodes
(internal nodes and only 2 leaves) is $(a,b)$.

In order to hold the information in the internal nodes
each internal node should maintain the following variables:
\begin{itemize}
\item[$l$]
A counter storing the difference between the number of
rectangles that were opened and that were closed
since the last traversal of the left child of $v$
and that cover the interval spanned by that child.
\item[$r$]
A counter storing the difference between the number of
rectangles that were opened and that were closed
since the last traversal of the right child of $v$
and that cover the interval spanned by that child.
\item[$l_m$]
A counter storing the maximum value of~$l$
since the last traversal of that child.
\item[$r_m$]
A counter storing the maximum value of~$r$
since the last traversal of that child.
\end{itemize}

\noindent
A leaf, on the other hand, holds two variables:
\begin{itemize}
\item[$c$]
The coverage of the associated interval during the sweep
at the point the leaf was traversed for the last time. 
\item[$c_m$]
The maximum coverage of the associated
interval during the sweep from the start until the leaf was traversed for the last time. 
\end{itemize}

\noindent
In relation to these values we define the following functions:
\[
 t(v) =  \left\{ \begin{array}{ll}
        u.l + t(u) & \mbox{if $v$ is the left child of $u$} \\
        u.r + t(u) & \mbox{if $v$ is the right child of $u$} \\
        0 & \mbox{if $v$ is the root} \\
        \end{array} \right.
,
\]

\[
 t_m(v) =  \left\{ \begin{array}{ll}
        max(u.l_m, u.l + t_m(u)) & \mbox{if $v$ is the left child of $u$} \\
        max(u.r_m, u.r + t_m(u)) & \mbox{if $v$ is the right child of $u$} \\
        0 & \mbox{if $v$ is the root} \\
        \end{array} \right.
 .
\]

\noindent
At any point of the sweep the following two invariants hold for every leaf $\ell$
and its associated interval $I$:
\begin{itemize}
    \item
    The current coverage of $I$ is: $\ell.c + t(\ell)$.
    \item
    The maximum coverage of $I$ that was seen so far is: $\max(\ell.c_m, \ell.c + t_m(\ell))$.
\end{itemize}

\noindent
Updating the structure with an event is done as follows:
Let $I$ be the $x$-interval spanned by the processed rectangle creating the event.
Depending on whether the rectangle starts or ends,
we set a variable $d=1$ or $d=-1$, respectively.
We follow the two search paths to the leftmost leaf and the rightmost leaf that are covered by $I$.
In the beginning the two paths are joined until they split, for every node $w$ on this path
(including the split node) we can ignore~$d$ and simply
update the tuple $(w.l,w.r,w.l_m,w.r_m)$ using $t(w)$ and $t_m(w)$
according to the invariants stated above. Note that this process needs to clear the
corresponding values in the parent node as otherwise the invariants would be violated.%
\footnote{Notice that using $t(w)$ and $t_m(w)$ here
takes constant time since we only need to access the parent
node as all previous nodes on the path towards the root are already processed.}
After the split the paths are processed separately.
We discuss here the left path, the behavior for the right path is symmetric.
Let $v$ be a node on the left path.
As long as $v$ is not a leaf we update $(v.l,v.r,v.l_m,v.r_m)$ as usual.
However, if the path continues to the left we also have to incorporate $d$ into
$v.r$ and $v.r_m$ as the subtree to the right is covered by $I$.
If $v$ is a leaf we simply update~$v.c$ and~$v.c_m$ using $t(v), t_m(v)$ and $d$.
A more detailed description (including pseudo code) can be found in~\cite{as-cdaaa-13}.
In total, this process takes $O(\log n)$ time.

Finally, in order to find the maximal number of rectangles covering
an interval one last propagation from root to leaves is needed, such that
all $l, r, l_m, r_m$ values of internal nodes are cleared. This is done using one
traversal on $T$.
Now, the maximal number of rectangles covering an interval
is the maximal~$c_m$ of all leaves of $T$.

Clearly, the running time of the algorithm is $O(n\log n)$,
since constructing the tree and sorting the $y$-events takes
$O(n\log n)$ time. Updating each of the~$2n$ $y$-events takes $O(\log n)$
time, and the final propagation of values to the leaves takes $O(n)$ time.
The algorithm uses $O(n)$ space.

}
{
} 

\clearpage

\section{Bijection between the Trapezoids in \TcollReduc and the Rectangles in \RcollReduc }
\label{appendix_reduction_bijection}

We devise here a proof for Theorem~\ref{theorem:reduction_correctness},
claiming that the mapping from~\TcollReduc to~\RcollReduc,
presented in Definition~\ref{def:reduction}, is bijective.

Given a subdivision, one can partition the plane into
vertical slabs by passing a vertical line through every endpoint
of the subdivision, and then partition each slab into regions by
intersecting it with all the curves in the subdivision.
This defines a decomposition of the plane into at most~$2(n+1)^2$
regions (see~\cite{CG-alg-app}, for example).

\begin{lemma}
\label{lemma:regions_bijection}
Let Regions(arr) denote the collection of regions
of an arrangement arr, as defined above.
For any region $a_t \in$ Regions(\Arr{\TcollReduc})
let $a_r \in$ Regions(\Arr{\RcollReduc}) be the
rectangular region corresponding to $a_t$.
The collection Regions(\Arr{\RcollReduc})
of all such rectangular regions spans the plane.
\end{lemma}

\begin{proof}
Trivial. The slabs remain the same and
within each slab the rectangular regions remain adjacent.
\myqed
\end{proof}

\begin{lemma}
\label{lemma:reduction_from_trpz_to_rect}
Let $a_t \in$ Regions(\Arr{\TcollReduc}) be a region
and let $a_r \in$ Regions(\Arr{\RcollReduc}) be the
rectangular region corresponding to $a_t$.
The number of rectangles in \RcollReduc that cover $a_r$
is at least the number of trapezoids in \TcollReduc that cover $a_t$.
In other words,
for every $\tred \in \TcollReduc$ that covers $a_t$
its corresponding rectangle $\rred \in \RcollReduc$ covers $a_r$.
\end{lemma}

\begin{proof}
Let $\{\tred_1,\tred_2,...,\tred_m\}\subseteq \TcollReduc$
be the set of trapezoids, ordered by creation time,
such that for every $i\in\{1,...,m\}$, $\tred_i$ covers $a_t$.
Let $\{\rred_1,\rred_2,...,\rred_m\}\subseteq \RcollReduc$
be the set of corresponding
rectangles, such that $\rred_i$ corresponds to $\tred_i$ for
$i \in \{1,...,m\}$.
For any $\tred_i$, since $\tred_i$ covers $a_t$
we get that $x\text{-range}(a_t)\subseteq x\text{-range}(\tred_i)$.
By Definition~\ref{def:reduction} the
$x$-ranges remain the same after the reduction,
and therefore $x\text{-range}(a_r) \subseteq x\text{-range}(\rred_i)$.
Since $\tred_i$ covers $a_t$ then we also get that
in the shared $x$-range top$(\tred_i)$ is above or on top$(a_t)$
and bottom$(\tred_i)$ is below or on bottom$(a_t)$.
According to Definition~\ref{def:reduction},
it immediately follows that
Rank$($top$(\tred_i)) \geq$ Rank$($top$(a_t))$.
In other words, top$(\rred_i)$ is above or on top$(a_r)$.
Similarly, bottom$(\rred_i)$ is below or on bottom$(a_r)$.
We conclude that $\rred_i$ covers~$a_r$.
\myqed
\end{proof}

\begin{lemma}
\label{lemma:reduction_from_rect_to_trpz}
Let $a_r\in$ Regions$(\Arr{\RcollReduc})$ be a rectangular region,
whose corresponding region is $a_t \in$ Regions$(\Arr{\TcollReduc})$.
The number of trapezoids in \TcollReduc that cover $a_t$
is at least the number of rectangles in \RcollReduc that cover $a_r$.
In other words,
for every $\rred \in \RcollReduc$ that covers $a_r$
its corresponding trapezoid $\tred \in \TcollReduc$ covers $a_t$.
\end{lemma}

\begin{proof}
Let $\{\rred_1,\rred_2,...,\rred_m\}\subseteq \RcollReduc$
be the set of rectangles,
such that for every $i\in\{1,...,m\}$, $\rred_i$ covers $a_r$.
Let $\{\tred_1,\tred_2,...,\tred_m\}\subseteq \TcollReduc$
be the set of corresponding
trapezoids, such that $\tred_i$ corresponds to $\rred_i$ for
$i \in \{1,...,m\}$.
Proving that for any $i \in \{1,...,m\}$, $\tred_i$ covers $a_t$,
is done symmetrically to the proof of Lemma~\ref{lemma:reduction_from_trpz_to_rect}.
\myqed
\end{proof}

Combining Lemma~\ref{lemma:reduction_from_trpz_to_rect} and
Lemma~\ref{lemma:reduction_from_rect_to_trpz}
we conclude that the number of trapezoids in \TcollReduc
that cover a region $a_t$ equals to the number
of rectangles in \RcollReduc that cover $a_r$, which is
the corresponding region to $a_t$.
The covering rectangles are the reduced trapezoids
in the set of trapezoids covering $a_t$.
Since both Regions$(\Arr{\TcollReduc})$
and Regions$(\Arr{\RcollReduc})$ span the plane
(Lemma~\ref{lemma:regions_bijection}),
we obtain Theorem~\ref{theorem:reduction_correctness}. 
\clearpage


\end{document}